\documentclass[10pt,journal,compsoc]{IEEEtran}
\def\BibTeX{{\rm B\kern-.05em{\sc i\kern-.025em b}\kern-.08em
    T\kern-.1667em\lower.7ex\hbox{E}\kern-.125emX}}
\usepackage{balance}

\usepackage{amsmath,amsfonts}
\usepackage{algorithmic}
\usepackage{array}
\usepackage{textcomp}
\usepackage{stfloats}
\usepackage{url}
\usepackage{verbatim}
\usepackage{graphicx}
\usepackage{cite}
\usepackage[utf8]{inputenc} 
\usepackage[T1]{fontenc}    
\usepackage{hyperref}       
\usepackage{url}            
\usepackage{booktabs}       
\usepackage{amsfonts}       
\usepackage{nicefrac}       
\usepackage{microtype}      

\usepackage{graphicx}
\usepackage{textcomp}
\usepackage{amssymb}

\usepackage{booktabs} 
\usepackage[table,xcdraw]{xcolor}
\usepackage{float}
\usepackage{multirow}
\usepackage{indentfirst}
\usepackage{multirow}
\usepackage{amsmath}
\usepackage{graphicx}
\usepackage[]{graphicx,indentfirst}
\usepackage{arydshln}
\usepackage{textcomp}
\usepackage{multicol}
\usepackage{bm}
\usepackage{url}
\usepackage{amsfonts}
\usepackage{mathrsfs}

\usepackage{algorithmic}
\usepackage[ruled,linesnumbered]{algorithm2e}

\newtheorem{proposition}{Proposition}

\newtheorem{lemma}{Lemma}
\newtheorem{proof}{Proof}
\newtheorem{theorem}{Theorem}

\usepackage{ulem}
\usepackage{cite}
\usepackage{amsmath,amssymb,amsfonts}
\usepackage{algorithmic}
\usepackage{textcomp}
\usepackage[mathscr]{eucal}
\usepackage{arydshln}
\usepackage{orcidlink}
\usepackage{booktabs}
\usepackage{subfigure}
\usepackage{subcaption}
\usepackage{threeparttable}
\usepackage{ragged2e}
\definecolor{myred}{HTML}{9B0303}
\definecolor{myblue}{HTML}{004873}

\begin{document}
\title{Variational Bayesian Personalized Ranking}
\author{Bin Liu~\textsuperscript{\orcidlink{0000-0002-1011-2909}} , 
	Xiaohong Liu~\textsuperscript{\orcidlink{0000-0001-6377-4730}}, 
        ~\IEEEmembership{Member,~IEEE},
	Qin Luo~\textsuperscript{\orcidlink{0009-0002-1425-4493}} ,
	Ziqiao Shang~\textsuperscript{\orcidlink{0009-0006-0468-7497}} ,
	Jielei Chu~\textsuperscript{\orcidlink{0000-0001-6232-5095}} ,
	Lin Ma~\textsuperscript{\orcidlink{0009-0005-1863-214X}} ,\\
	Zhaoyu Li~\textsuperscript{\orcidlink{0009-0002-1425-4493}} ,
	Fei Teng~\textsuperscript{\orcidlink{0000-0001-9535-7245}},
	Guangtao Zhai~\textsuperscript{\orcidlink{0000-0001-8165-9322}},~\IEEEmembership{Fellow,~IEEE} and
	Tianrui Li~\textsuperscript{\orcidlink{0000-0001-7780-104X}},~\IEEEmembership{Senior Member,~IEEE}
	\thanks{
		Corresponding authors: Xiaohong Liu, Tianrui Li.\\
		Bin Liu, Jielei Chu, Zhaoyu Li, Fei Teng, Tianrui Li are with School of Computing and Artificial Intelligence, Southwest Jiaotong University (SWJTU), Chengdu, China. (e-mail: \{binliu,jieleichu,lizhaoyu,fteng,trli\}@swjtu.edu.cn) 
		
		Xiaohong Liu, Guangtao Zhai are with School of Electronic Information and Electrical Engineering, Shanghai Jiao Tong University (SJTU), Shanghai, China. (e-mail: \{xiaohongliu, zhaiguangtao\}@sjtu.edu.cn) 
		
		Qin Luo, Ziqiao Shang are with School of Electronic Information and Communications, Huazhong University of Science and Technology (HUST), Wuhan, China. (e-mail: luo\_qin@hust.edu.cn, ziqiaosh@gmail.com)
		
		Lin Ma is with School of Transportation and Logistics, Southwest Jiaotong University, Chengdu, China. (e-mail:malin@swjtu.edu.cn)
		
		This research was supported by the National Natural Science Foundation of China (Nos. U2468207, 62572317, 62276218, 62272398 62176221, 61572407, 52472333), the Fundamental Research Funds for the Central Universities No.2682025CX082, Sichuan Science and Technology Program (Nos. 2024NSFTD0036, 2024ZHCG0166), Science and Technology R\&D Program of China Railway Group Limited (2023-Major-23).}} 

\markboth{IEEE TPAMI, 2026. doi:~\url{10.1109/TPAMI.2026.3672705}}{}

\IEEEtitleabstractindextext{%
\justifying

\begin{abstract}
Pairwise learning underpins implicit collaborative filtering, yet its effectiveness is often hindered by sparse supervision, noisy interactions, and popularity-driven exposure bias. In this paper, we propose Variational Bayesian Personalized Ranking (VarBPR),
a tractable variational framework for implicit-feedback pairwise learning that offers principled exposure controllability and theoretical
interpretability. VarBPR reformulates pairwise learning as variational inference over discrete latent indexing variables, explicitly modeling
noise and indexing uncertainty, and divides training into two stages: variational inference, which solve variational posteriors, and variational
learning, which updates model parameters based on these posteriors. In the variational inference stage, we develop a variational
formulation that integrates preference alignment, denoising, and popularity debiasing under a unified ELBO/regularization objective,
deriving closed-form posteriors with clear control semantics: the prior encodes a target exposure pattern, while temperature/regularization
strength controls posterior-prior adherence. As a result, exposure controllability becomes an endogenous and interpretable outcome
of variational inference. In the variational learning stage, we propose a posterior-compression objective that reduces the ideal ELBO’s
computational complexity from polynomial to linear, with the approximation justified by an explicit Jensen-gap upper bound. Theoretically,
we provide interpretable generalization guarantees by identifying a structural error component and revealing the opportunity cost of
prioritizing certain exposure patterns (e.g., long-tail), offering a concrete analytical lens for designing controllable recommender systems.
Empirically, We validate VarBPR across popular backbones; it demonstrates consistent gains in ranking accuracy, enables controlled longtail exposure, and preserves the linear-time complexity of BPR. Code and data are available at: \url{https://github.com/liubin06/VariationalBPR}.

\end{abstract}

\begin{IEEEkeywords}
Bayesian personalized ranking, variational inference, pairwise learning, self-supervised learning, contrastive learning, controllable recommendation. 
\end{IEEEkeywords}}

\maketitle
\IEEEdisplaynontitleabstractindextext
\IEEEpeerreviewmaketitle

\section{Introduction}
\IEEEPARstart{P}{airwise} learning is a fundamental paradigm for implicit collaborative filtering and has been widely adopted across a broad range of applications, including e-commerce, social media, and online entertainment, where user feedback is typically implicit , and personalized content delivery is required~\cite{7839995,Hou:ijcai,lin2024temporally,lai2024knowledge}. 
Implicit feedback (e.g., clicks, purchases, and likes), which represents the most common form of interaction data~\cite{pan2009mind}, provides only indirect preference signals, unlike explicit ratings that directly indicate user preferences. 
As a result, traditional supervised learning paradigms, which minimize discrepancies between predicted and actual ratings, are not directly applicable in this context. 
In response, Bayesian Personalized Ranking (BPR)~\cite{Steffen:2009:UAI} introduced the shift from rating prediction to ranking optimization. 
By maximizing the likelihood of correctly ordered item pairs, BPR optimizes a pairwise ranking objective aligned with the Area Under the Curve (AUC), thereby better matching the training objective to ranking-based evaluation. 
BPR has since become a widely used ranking objective for training recommender models~\cite{Wang:2019:SIGIR,Xiangnan:2020:SIGIR,10158930,ma2024tail} and is often combined with self-supervised learning to further improve performance~\cite{10158930,SSR:2023:TKDE}. 
More recently, pairwise ranking principles related to BPR have informed preference modeling in reinforcement learning from human feedback (RLHF), where they are used to align large language models (LLMs) with human preferences, demonstrating the broad applicability of ranking-based techniques beyond recommender systems~\cite{10.5555/3666122.3668460,Yang2024RegularizingHS}.

\begin{figure}[!t]
	\centering
	\includegraphics[width=0.47\textwidth]{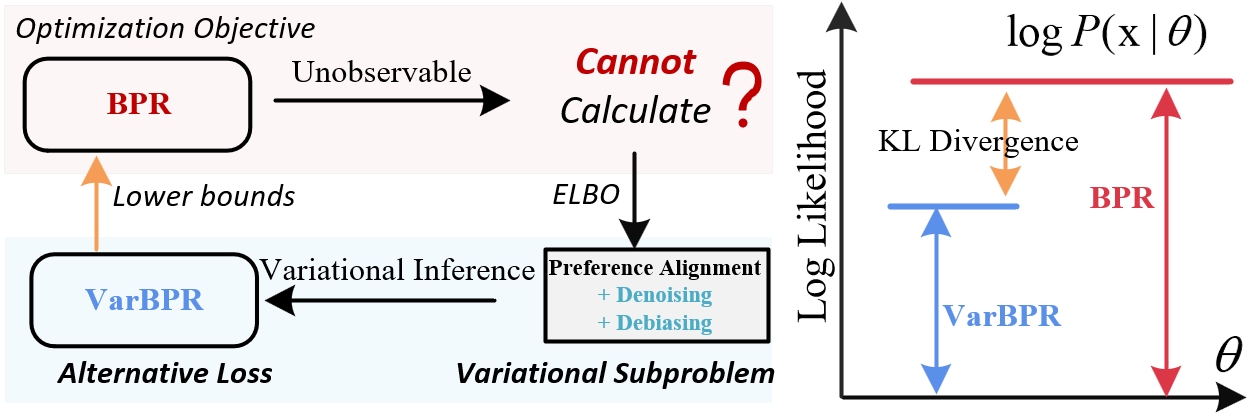}
	\caption{Illustration of VarBPR. Unobservable preferences complicate the likelihood estimation. We establish a unified variational framework that integrates preference alignment, denoising, debiasing from ELBO with closed-form variational posteriors. This yields VarBPR objective for indirect likelihood maximization, featuring analytical inference procedures with explicit exposure  control mechanisms.}
	\label{fig:bound}
\end{figure}

Implicit feedback, as a form of incomplete data~\cite{Dempster:1977:RSS}, allows us to observe only part of user interactions, without revealing their actual preferences. BPR leverages the co-occurrence input~\cite{Liu:2021:TKDE} of user-item interactions to automatically label interacted items as positive examples and non-interacted items as negative examples. By introducing the assumption that users prefer positive items over negative ones (self-supervised signal)~\cite{SSR:2023:TKDE}, BPR constructs ordered pairs and performs maximum likelihood estimation on these observed pairs. This automated labeling reduces the dependence on intensive manual annotation, making it feasible to learn rankings from user behavior data. However, this approach also introduces the problem of false positive (FP) and false negative (FN) samples~\cite{wang2021denoising,Ding:2019:IJCAI,Bin:2023:ICDE,lai2024adaptive}. FP examples occur when users interact with but do not actually like an item (e.g., accidental clicks), while FN examples are items users prefer but have not interacted with yet.

\IEEEpubidadjcol

The inherent limitations of implicit feedback give rise to three fundamental challenges in recommender systems.
First, the parameter estimation challenge stems from the lack of explicit preference labels, which makes likelihood estimation intrinsically ill-posed.
Since implicit feedback provides only partial and exposure-dependent observations of user preference, the resulting missing data mechanism~\cite{Dempster:1977:RSS} complicates accurate estimation of pairwise preference likelihoods.
Second, the noise challenge arises from label contamination, including false positive interactions as well as unobserved but truly preferred false negatives.
Such noise violates the clean pairwise-comparison assumption underlying BPR, distorts the learning trajectory, and ultimately degrades ranking quality~\cite{Bin:2023:ICDE}.
Third, the exposure bias challenge reflects a systematic mismatch between observed interactions and the underlying preference distribution~\cite{chen2023bias,wang2022invariant,chen2021autodebias,lin2024recrec}.
 In particular, the exposure-driven interaction paradigm induces a systemic imbalance where unobserved false negative samples remain statistically underrepresented, while observed interactions become skewed toward a limited subset of dominant popular items. This imbalance initiates a self-reinforcing feedback loop that systematically amplifies popularity bias in recommendation outcomes~\cite{chen2023bias}.

Significant research has been devoted to improving implicit collaborative filtering for accurate preference ranking. 
On the front of parameter estimation, inverse propensity scoring (IPS)~\cite{schnabel2016recommendations} reweighting is a key statistical debiasing approach, aiming to yield unbiased loss estimates. However, their practical efficacy remains contingent upon precise propensity score estimation—a notoriously ill-posed problem that persists as an open challenge~\cite{saito2020asymmetric}. Recent work reveals that slight estimation errors can severely compromise its unbiasedness~\cite{cao2024practically}.
In dealing with noise, noise-label learning methods~\cite{song2022learning} attempt to differentiate clean from noisy samples by exploiting heuristics in neural networks' memorization dynamics~\cite{gao2022self} or gradient-fitting characteristics~\cite{wang2021denoising}. While these approaches modulate training to reduce noise propagation, their performance is often architecture-dependent, limiting their robustness in real-world applications.
Concerning popularity bias, prevalent approaches~\cite{Liupopdcl,zhang2022incorporating} suppress exposure of popular items via adaptive score calibration. However, score calibration is inherently heuristic and post-hoc, providing opaque control over the final exposure profile. 
Together, these limitations motivate the pursuit of a unified and lightweight alternative capable of handling practical challenges in parameter estimation, noise mitigation, and popularity debiasing.

To address these challenges, we propose Variational Bayesian Personalized Ranking (VarBPR). VarBPR reformulates pairwise learning as variational inference over discrete latent indexing variables that explicitly model the inherent noise and indexing uncertainty in implicit feedback, and it optimizes the evidence lower bound (ELBO), as illustrated in Fig.~\ref{fig:bound}.
 This formulation naturally separates training into two stages: variational inference, which solves for the variational posteriors, and variational learning, which updates model parameters given the inferred posteriors. 
Within this inference structure, VarBPR unifies three practical concerns in implicit collaborative filtering---preference alignment, noise reduction, and popularity debiasing (exposure control)---under a single ELBO/regularization objective, yielding closed-form optimal posteriors. Crucially, the variational inference mechanism provides explicit and interpretable control semantics for exposure: the prior encodes the desired exposure pattern, while temperature/regularization hyperparameters continuously control posterior--prior compliance,  enabling controllable long-tail exposure.
To enable large-scale learning, we introduce a posterior-compression objective that reduces the ideal ELBO computation  from polynomial to linear complexity, and we justify the approximation via an explicit Jensen-gap upper bound.
Theoretically, we provide interpretable generalization guarantees by identifying a structural error component induced by variational inference and analyzing its management strategy, thereby revealing the inherent opportunity cost of prioritizing certain exposure patterns. This analysis further yields a unified tuning principle centered on managing this cost and offers a concrete analytical lens for understanding and designing controllable recommender systems. Empirically, we validate VarBPR on several popular backbones, including MF~\cite{Koren:2009:Computer}, LightGCN~\cite{Xiangnan:2020:SIGIR}, and XSimGCL~\cite{10158930}. 
Across these models, VarBPR demonstrates consistent gains in ranking accuracy, enables controlled long-tail exposure.
The key contributions of this paper are as follows:
\vspace{5pt}
\begin{itemize}
\item \textbf{Unified variational framework.}  
We develop a variational inference framework that integrates preference alignment, popularity debiasing (exposure control), and noise reduction, yielding the novel VarBPR as the variational learning objective and providing a robust lower bound for ranking optimization.

\item \textbf{Analytical inference and endogenous exposure control mechanism.}  We derived closed-form posteriors.
Exposure controllability becomes an endogenous and interpretable outcome of the analytical variational inference process: the priors specify the desired exposure pattern, while the temperatures control the strength of posterior adherence to the prior.

\item \textbf{Generalization insight for exposure control.}  
We derive generalization bounds for VarBPR and reveal the inherent “opportunity cost” of prioritizing certain exposure patterns, establishing a unified tuning principle and providing a concrete analytical lens for designing controllable recommender systems.

\item \textbf{Effectiveness, controllability, and scalability.}  
Our experiments on popular backbones show that VarBPR consistently improves ranking accuracy, enables controlled long-tail exposure, and preserves the linear-time complexity of pairwise loss.
\end{itemize}





\section{Related Work}
\subsection{Pairwise Learning for Recommendation}
Early recommender systems were primarily designed to solve the rating prediction problem~\cite{hofmann2004latent,liu2009probabilistic}, where the goal was to predict users' ratings on items based on explicit feedback data, such as star ratings or reviews. While this approach worked well with explicit feedback, it encountered significant challenges when applied to implicit feedback, such as clicks, purchases, or views~\cite{pan2008one}. In these cases, directly optimizing for rating prediction does not accurately capture users' true preferences, as implicit feedback only indicates whether an interaction occurred, without revealing explicit satisfaction levels. To address this limitation, the seminal work on Bayesian Personalized Ranking (BPR)~\cite{Steffen:2009:UAI} introduced pairwise learning for recommendation, reformulating the task from rating prediction to ranking prediction. By focusing on ranking instead of direct prediction, BPR effectively models users' relative preferences between interacted and non-interacted items. This paradigm shift has since dominated the field of implicit collaborative filtering~\cite{yu2020collaborative}, leading many state-of-the-art recommendation models to adopt BPR as their primary optimization objective. Notable examples include NGCF~\cite{Wang:2019:SIGIR}, LightGCN~\cite{Xiangnan:2020:SIGIR}, and XSimGCL~\cite{10158930}. Building on BPR, researchers have proposed various enhancements to improve the model's performance~\cite{he2016vbpr,Yu:2018:CIKM,yu2020collaborative,Bin:TKDE:2024}. For instance, Group Bayesian Personalized Ranking (GBPR)~\cite{Weike:2013:IJCAI} extends BPR by incorporating group preferences, allowing it to better model collective user behavior. Meanwhile, MPR~\cite{Yu:2018:CIKM} introduces the integration of negative feedback to learn more fine-grained user preferences. However, despite these advancements, many of these improvements fail to systematically address the noise introduced by false positive (FP) and false negative (FN) samples. The lack of effective handling of these noisy signals limits the overall accuracy and reliability of the recommendations, particularly in the context of implicit feedback data.

\subsection{Self-supervised Learning for Recommendation} 
Self-supervised learning (SSL)~\cite{BYOL:2020:NIPS,wu:2023:TKDE,gsl:2023:TKDE} has gained significant traction in recommendation systems, giving rise to what is now referred to as Self-Supervised Recommendation (SSR)~\cite{SSR:2023:TKDE}. SSR encompasses a range of techniques, including generative models~\cite{sun2019bert4rec,geng2022recommendation} and contrastive models~\cite{xu2024cmclrec,zhou2024cllp}. Among these, Bayesian Personalized Ranking (BPR) is closely related to contrastive SSR~\cite{Bin:TKDE:2024}. 
When it comes to data labeling, BPR automatically assigns positive or negative labels to items based on the co-occurrence of user-item interactions~\cite{Liu:2021:TKDE}. BPR assumes that users prefer positive items over negative ones, pulling positive items closer to the user (the anchor) while pushing negative items further away~\cite{Bin:TKDE:2024}. This strategy aligns with the optimization of ranking metrics such as the Area Under the Curve (AUC)~\cite{Steffen:2009:UAI} and Normalized Discounted Cumulative Gain (NDCG)~\cite{Jiancan:2022:arxiv}, both of which are critical to the ranking task.
Furthermore, the BPR objective is intrinsically connected to contrastive learning. Specifically, it is equivalent to the InfoNCE~\cite{Oord:2018:arxiv} loss with exactly one negative sample~\cite{Bin:TKDE:2024}. Another prominent approach within SSR is graph contrastive learning~\cite{wu2021self,zhang2024temporal}, which utilizes self-supervised signals derived from slight perturbations in graph structures~\cite{Xia:2022:ICML} to maintain semantic consistency. Despite the distinctiveness of graph contrastive learning, it is often combined with BPR for optimization~\cite{lightgcl:2023:ICLR,10158930}, underscoring the complementary nature of these techniques.
The successful application of BPR in implicit collaborative filtering highlights the value and potential of SSR in enhancing recommendation models. However, the effectiveness of BPR and its extensions is hindered by key challenges in implicit feedback data—namely, sparse supervision, noise, and popularity-induced exposure bias—which continues to be a vital area of research for improving the robustness and performance of these approaches.

\section{Methodology}
\subsection{Pairwise Learning as MLE}
Given a user set \(\mathcal{U}\), an item set \(\mathcal{I}\), and their interactions, the objective of personalized recommendation is to learn an embedding function \(f_\theta\) that maps a user \(u \in \mathcal{U}\) to a user feature representation \(\mathbf{u} = f(u)\) and an item \(i \in \mathcal{I}\) to an item feature representation \(\mathbf{i} = f(i)\). These feature representations should effectively capture the characteristics of the respective user or item. The user's preference for an item is modeled by the similarity between their feature representations. The top-\(k\) recommendations for a user \(u\) are determined by selecting the \(k\) items most similar to \(u\)'s feature representation.

For each user \(u\), let \(\mathcal{I}_u^+\) denote the set of items with which the user has interacted, and \(\mathcal{I}_u^-\) the set of items the user has not interacted with. The training set can then be defined as:
\begin{eqnarray}
   \mathcal{D} = \{(u,i,j) \mid u \in \mathcal{U}, i \in \mathcal{I}_u^+, j \in \mathcal{I}_u^- \}. 
\end{eqnarray}
Each training triplet \((u,i,j)\) represents an observed user preference order \(i >_u j\), indicating that user \(u\) prefers the positive item \(i \in \mathcal{I}_u^+\) over the negative item \(j \in \mathcal{I}_u^-\)~\cite{Steffen:2009:UAI}.

Given the preference order \(i >_u j\), the likelihood is modeled using the Bradley-Terry function~\cite{Bradley:1952:Biometrika}:
\begin{eqnarray}\label{eq:bradlyterry}
    P(i >_u j | \theta) = \sigma(\hat{x}_{ui} - \hat{x}_{uj}),
\end{eqnarray}
where \(\sigma(\cdot)\) is the sigmoid function, and \(\hat{x}_{ui} = \text{sim}(\mathbf{u}, \mathbf{i})\) represents the similarity between the user and the item, typically measured using inner product similarity. The maximum likelihood estimation is then achieved by minimizing the Bayesian Personalized Ranking (BPR) loss:
\begin{align}\label{eq:bpr}
    \mathcal{L}_\text{BPR} =  - \frac{1}{|\mathcal{D}|} \sum_{(u,i,j)} \ln \sigma (\hat{x}_{ui} - \hat{x}_{uj}).
\end{align}
Thus, pairwise learning can be framed as a maximum likelihood estimation (MLE) problem, where the goal is to find the optimal model parameters \(\theta\) that maximize the likelihood of the preference outcome corresponding to  all ordered pairs in the dataset \(\mathcal{D}\). In practice, regularization terms are introduced to prevent overfitting, often interpreted as the logarithm of Gaussian prior densities. Consequently, the BPR estimator in Eq.~\eqref{eq:bpr} can be viewed through the lens of posterior probability. In this context, we do not differentiate between the interpretations of BPR that arise due to the semantics of the regularization term. 
However, the inherent noise in implicit feedback poses a significant challenge to reliable preference learning.

We next revisit pairwise learning from a variational perspective and show how it leads to the proposed VarBPR formulation.

\subsection{Variational Lower Bound}

To account for the noisy nature of implicit feedback, we enrich each BPR triplet \((u,i,j)\) into an enriched interaction 
\(\mathbf{x} = (u,i_{1:M},j_{1:N})\) that denotes the corresponding preference outcome, where \(\{i_1,\dots,i_M\}\subset\mathcal I_u^+\) and \(\{j_1,\dots,j_N\}\subset\mathcal I_u^-\).
For a given enriched interaction \(\mathbf x\), we introduce a discrete latent index variable to explicitly capture the inherent noise and the indexing uncertainty
\begin{equation}
\mathbf{h}=[h^+,h^-],
\end{equation}
where \(h^+\in\{1,\dots,M\}\) indexes the positive sample that best reflects the user's underlying positive interest, and \(h^-\in\{1,\dots,N\}\) indexes the negative sample that best reflects the user's underlying negative interest.

We denote the  prior over \(\mathbf h\) by
\(P(\mathbf h)=P^+(h^+)\,P^-(h^-)\),
and the variational distribution by
\(q(\mathbf h)=q^+(h^+)\,q^-(h^-)\).
The likelihood admits following ELBO-KL decomposition:
\begin{align}
	\ln P(\mathbf{x} \mid \theta)
	&= \sum_\mathbf{h} q(\mathbf{h}) \ln P(\mathbf{x} \mid \theta) \nonumber \\
	&= \sum_\mathbf{h} q(\mathbf{h}) \ln \frac{P(\mathbf{x},\mathbf{h} \mid \theta)}{P(\mathbf{h} \mid \mathbf{x},\theta)} \nonumber \\
	&= \sum_\mathbf{h} q(\mathbf{h}) \ln \frac{P(\mathbf{x},\mathbf{h} \mid \theta)}{q(\mathbf{h})} 
	   + \sum_\mathbf{h} q(\mathbf{h}) \ln \frac{q(\mathbf{h})}{P(\mathbf{h} \mid \mathbf{x},\theta)} \nonumber \\
	&= \mathcal{J}(\theta) + \mathrm{KL}\big(q(\mathbf{h}) \,\Vert\, P(\mathbf{h} \mid \mathbf{x},\theta)\big), \label{eq:elbo-kl}
\end{align}
where
\begin{align}
  \mathcal{J}(\theta)
&= \sum_{\mathbf{h}} q(\mathbf{h}) \ln \frac{P(\mathbf{x},\mathbf{h} \mid \theta)}{q(\mathbf{h})}  \nonumber \\
&= \mathbb E_{q{(\mathbf{h})}} \big[\ln P(\mathbf{x}|\mathbf{h},\theta)\big] + H(q) - H(q,P(\mathbf h)) \label{eq:elbo}
\end{align}
is the evidence lower bound (ELBO). Here \(H(q)\) denotes the entropy of the variational distribution, and
\(H(q,P(\mathbf h))\) denotes the cross-entropy of variational distribution \(q(\mathbf h)\) with respect to the prior \(P(\mathbf h)\).
Since the KL term in Eq.~\eqref{eq:elbo-kl} is nonnegative, \(\mathcal{J}(\theta)\) lower bounds \(\ln P(\mathbf{x}\mid\theta)\) .

We now specify the generative model.
Given positive index variable \(h^+=m\) and negative index variable \(h^-=n\), the latent mapping selects the clean pair \((i_m,j_n)\) from the enriched candidate set $\mathbf{x}$, and the preference outcome is modeled in the same way as BPR:
\begin{small}
\begin{equation}\label{eq:cp}
P(\mathbf{x} \mid \mathbf{h},\theta) 
	= \sigma\Big(\langle \mathbf{u}, \mathbf{i}^+_m\rangle
	                  - \langle \mathbf{u}, \mathbf{j}^-_n\rangle\Big),
\end{equation}
\end{small}
\par\noindent where \(\mathbf u=f_\theta(u)\), \(\mathbf i_m^+=f_\theta(i_m)\), and \(\mathbf j_n^-=f_\theta(j_n)\).

For notational convenience, we represent the variational distribution and the priors in vector form:
\begin{align*}
\boldsymbol{\alpha}&=(\alpha_1,\dots,\alpha_M),\ \alpha_m \triangleq q^+(h^+=m), \\
\boldsymbol{\beta}&=(\beta_1,\dots,\beta_N), ~\beta_n \triangleq q^-(h^-=n),  
\end{align*}
and 
\begin{align*}
   \boldsymbol{\pi}^+&=(\pi_1^+,\dots,\pi_M^+),\ \pi_m^+ \triangleq P^+(h^+=m),\\
\boldsymbol{\pi}^-&=(\pi_1^-,\dots,\pi_N^-),\ \pi_n^- \triangleq P^-(h^-=n). 
\end{align*}

Substituting Eq.~\eqref{eq:cp} into Eq.~\eqref{eq:elbo} and using
\(H(q(\mathbf h))=H(\boldsymbol{\alpha})+H(\boldsymbol{\beta}) \) and
\(H(q(\mathbf h),P(\mathbf h))=H(\boldsymbol{\alpha},\boldsymbol{\pi}^+)+H(\boldsymbol{\beta},\boldsymbol{\pi}^-)\), we obtain the final variational lower bound of the form
\begin{small}
  \begin{align}
\mathcal{J}(\theta)
= \sum_{m=1}^M\sum_{n=1}^N &\alpha_m\beta_n\log \sigma\Big(\langle \mathbf{u}, \mathbf{i}^+_m\rangle - \langle \mathbf{u}, \mathbf{j}^-_n\rangle\Big)\nonumber\\
&+ H(\boldsymbol{\alpha}) - H(\boldsymbol{\alpha},\boldsymbol{\pi}^+) 
+ H(\boldsymbol{\beta}) - H(\boldsymbol{\beta},\boldsymbol{\pi}^-).\label{eq:elbo1}
\end{align}  
\end{small}

\subsection{Variational Inference}
This subsection describes the variational inference step, i.e., how we solve for the variational posteriors
\(\boldsymbol\alpha=(\alpha_1,\dots,\alpha_M)\) and
\(\boldsymbol\beta=(\beta_1,\dots,\beta_N)\). Throughout, \(\boldsymbol\alpha\) and \(\boldsymbol\beta\) lie on the probability simplices
\(\Delta^M=\{\boldsymbol\alpha:\alpha_m\ge 0,\ \sum_{m=1}^M\alpha_m=1\}\) and
\(\Delta^N=\{\boldsymbol\beta:\beta_n\ge 0,\ \sum_{n=1}^N\beta_n=1\}\), respectively.
In principle, variational inference selects \((\boldsymbol\alpha,\boldsymbol\beta)\) by maximizing the ELBO in Eq.~\eqref{eq:elbo1}. However, the ELBO likelihood term
\(\sum_{m,n}\alpha_m\beta_n\log\sigma(\gamma_{mn})\), with margin
\(\gamma_{mn}=\langle\mathbf u,\mathbf i_m\rangle-\langle\mathbf u,\mathbf j_n\rangle\),
couples \(\boldsymbol\alpha\) and \(\boldsymbol\beta\) through the nonlinear \(\log\sigma(\cdot)\) and yields a constrained, non-separable optimization problem. As a result, a direct E-step typically requires an inner-loop iterative solver per enriched interaction, which increases computation and may introduce additional optimization instability.

\textbf{Linearization surrogate}. Our key modeling signal is margin-based ranking: we want to increase user-positive similarity and decrease user-negative similarity, as encoded by the pairwise margin $\gamma_{mn} =\langle \mathbf u,\mathbf i_m \rangle-\langle \mathbf u,\mathbf j_n\rangle $. Since AUC depends only on the ordering induced by \(\gamma\), and \(\log\sigma(\gamma)\) is strictly increasing in \(\gamma\), replacing \(\log\sigma(\gamma_{mn})\) by its linear margin surrogate preserves ranking consistency.More importantly, the log-likelihood admits the Maclaurin expansion
\(\log\sigma(\gamma)= -\log 2 + \gamma/2 + \epsilon(\gamma)\),
with approximation error bounded by \(|\epsilon(\gamma)| \le \gamma^2/8\).
Therefore, maximizing the expected margin captures the first-order ranking signal of the ELBO likelihood term while incurring a controlled second-order approximation error. We therefore adopt the following margin-based variational inference objective.
\begin{small}
\begin{align}\label{eq:marginbased-vi}
\mathcal V(\boldsymbol\alpha,\boldsymbol\beta)
=
&\sum_{m=1}^M\sum_{n=1}^N \alpha_m\beta_n\Big(\langle \mathbf{u}, \mathbf{i}_m\rangle - \langle \mathbf{u}, \mathbf{j}_n\rangle\Big)
\nonumber \\
&+H(\boldsymbol\alpha)-H(\boldsymbol\alpha,\boldsymbol\pi^+)
\;+\;
H(\boldsymbol\beta)-H(\boldsymbol\beta,\boldsymbol\pi^-).
\end{align}   
\end{small}
\par\noindent Crucially, the bilinear summation 
\(\sum_{m}\sum_n\alpha_m\beta_n\big(\langle \mathbf{u}, \mathbf{i}_m\rangle - \langle \mathbf{u}, \mathbf{j}_n\rangle\big)
=\sum_m \alpha_m\langle\mathbf u,\mathbf i_m\rangle-\sum_n \beta_n\langle\mathbf u,\mathbf j_n\rangle\) in Eq.~\eqref{eq:marginbased-vi} separates into a positive term depending only on \(\boldsymbol\alpha\) and a negative term depending only on \(\boldsymbol\beta\).
By further introducing coefficients 
$c_{\text{pos}}, c_{\text{neg}}$ to control the strength of the prior-matching effect, the original problem decomposes into two independent variational subproblems. 

\textbf{Variational subproblems}. On the positive side:
\begin{small}
\begin{equation}\label{eq:var-pos}
\mathcal{V}^+:\max_{\boldsymbol\alpha\in\Delta^M}
  \Big\{
    \underbrace{\sum_{m=1}^M\alpha_m\langle\mathbf{u},\mathbf{i}_m \rangle}_{\text{preference alignment}}
    + \underbrace{c_{\text{pos}}H(\boldsymbol{\alpha})}_{\text{denoising}}
    - \underbrace{c_{\text{pos}}H(\boldsymbol{\alpha}, \boldsymbol{\pi}^+)}_{\text{popularity debiasing}}
  \Big\},
\end{equation}
\end{small}
\par\noindent Variational subproblem~\eqref{eq:var-pos} can be viewed as a regularized maximum alignment problem: the $\sum_{m=1}^M\alpha_m\langle\mathbf{u},\mathbf{i}_m \rangle$ term encourages the variational distribution $\boldsymbol{\alpha}$ to maximize preference alignment. The entropy term \(H(\boldsymbol{\alpha})\) promotes a smooth  assignment of variational distributuion $\boldsymbol{\alpha}$ and
prevents over-reliance on any single, potentially noisy sample
(denoising). While the cross-entropy term
\(H(\boldsymbol{\alpha},\boldsymbol\pi^+)\) compares the variational distribution
\(\boldsymbol\alpha\) with the desired exposure prior \(\boldsymbol\pi^+\).
Because it appears in the objective as
\(-H(\boldsymbol\alpha,\boldsymbol\pi^+)\), it penalizes
assignments that deviate from this exposure profile (popularity debiasing). The hyperparameter \(c_{\text{pos}}>0\) (introduced herein) controls the overall strength of these regularization effects, allowing the inference step to balance these effects in a controlled manner.

On the negative side:
\begin{small}
\begin{equation}\label{eq:var-neg}
\mathcal{V}^-:\max_{\boldsymbol\beta\in\Delta^N}
  \big\{
    \underbrace{-\sum_{n=1}^N\beta_n\langle \mathbf{u}, \mathbf{j}_n \rangle}_{\text{preference alignment}}
    + \underbrace{c_{\text{neg}} H(\boldsymbol{\beta})}_{\text{denoising}}-\underbrace{c_{\text{neg}} H(\boldsymbol{\beta}, \boldsymbol{\pi}^-)}_{\text{popularity debiasing}}
  \big\},
\end{equation}
\end{small}
\par\noindent
Analogously, the $-\sum_{n=1}^N\beta_n\langle \mathbf{u}, \mathbf{j}_n \rangle$ term encourages the variational distribution $\boldsymbol{\beta}$ to maximize preference alignment. \(H(\boldsymbol{\beta})\) encourages a smooth assignment (denoising), while the term
\(-c_{\text{neg}}H(\boldsymbol{\beta},\boldsymbol{\pi}^-)\) penalizes deviations from the suppression prior \(\boldsymbol{\pi}^-\), thereby modulating the suppression pattern. The hyperparameter \(c_{\text{neg}}>0\) (introduced herein) controls the overall strength of regularization effects.

\vspace{5pt}
\begin{lemma}[Variational optimum]\label{lemma:solution}
    The closed-form optima of Eqs.~\eqref{eq:var-pos} and \eqref{eq:var-neg} are
\begin{align}
\alpha_m 
&= \frac{\pi_m ^+\exp\big(\mathbf{u}^\mathsf{T} \mathbf{i}_m / c_{\text{pos}}\big)}
       {\sum_{r=1}^{M} \pi_r ^+\exp\big(\mathbf{u}^\mathsf{T} \mathbf{i}_r / c_{\text{pos}}\big)},\label{eq:var:alpha}\\
\beta_n 
&= \frac{\pi_n ^-\exp\big(-\mathbf{u}^\mathsf{T} \mathbf{j}_n / c_{\text{neg}}\big)}
       {\sum_{r=1}^{N} \pi_r ^-\exp\big(-\mathbf{u}^\mathsf{T} \mathbf{j}_r / c_{\text{neg}}\big)}. \label{eq:var:beta}    
\end{align}
\end{lemma}

\begin{proof}
On the positive side, introducing a Lagrange multiplier \(\lambda\) for the variational distribution constraint \(\sum_m \alpha_m = 1\) in Eq.~\eqref{eq:var-pos}, we obtain the Lagrangian
\begin{small}
\begin{eqnarray}
\mathcal V^+
= \sum_{m=1}^M \alpha_m \langle \mathbf u,\mathbf i_m\rangle
  - c_{\text{pos}} \sum_{m=1}^M \alpha_m \log \frac{\alpha_m}{\pi_m^+}
  + \lambda\Big(\sum_{m=1}^M \alpha_m - 1\Big). \nonumber
\end{eqnarray}
\end{small}
Taking derivatives with respect to each \(\alpha_m\) and setting them to zero gives
\[
\frac{\partial \mathcal V^+(\boldsymbol\alpha,\lambda)}{\partial \alpha_m}
= \langle \mathbf u,\mathbf i_m\rangle
  - c_{\text{pos}}\Big(\log\frac{\alpha_m}{\pi_m^+} + 1\Big)
  + \lambda = 0.
\]
Rearranging yields
\[
\log\frac{\alpha_m}{\pi_m^+}
= \frac{1}{c_{\text{pos}}}\big(\langle \mathbf u,\mathbf i_m\rangle + \lambda - c_{\text{pos}}\big),
\]
where \(\lambda - c_{\text{pos}}\) is independent of \(m\).
Exponentiating yields:
\[\alpha_m \propto \pi_m ^+\exp\big(\frac{\mathbf{u}^\mathsf{T}\mathbf{i}_m }{c_{\text{pos}}}\big) \]
Normalizing over \(m\) yields the claimed equality in Eq.~\eqref{eq:var:alpha}. Symmetrically, the same procedure applied to the negative side is omitted for brevity.
\end{proof}


\textbf{Tractable inference.}  
Lemma~\ref{lemma:solution} shows that the variational posteriors \((\boldsymbol\alpha,\boldsymbol\beta)\) admit closed-form solutions. This analytical form provides three key advantages: (i) \textit{interpretability}, as it allows for a direct understanding of the posterior derivation and the role of each parameter in the inference process, thereby providing an analytical interface to the inference procedure; (ii) \textit{computational efficiency}, by enabling constant-time inference; and (iii) \textit{numerical stability}, ensuring reliable posterior estimates without the convergence risks inherent in iterative optimization methods. Moreover, when solving Eqs.~\eqref{eq:var:alpha}--\eqref{eq:var:beta}, no additional bag-wise normalization of $(\boldsymbol{\pi}^+,\boldsymbol{\pi}^-)$ is required, since $(\boldsymbol{\alpha},\boldsymbol{\beta})$ are normalized over the simplex, and their optimal solutions remain invariant under any common scaling of the constrained prior within a bag.

\textbf{Exposure control mechanism.}  
The variational posteriors \(\boldsymbol\alpha\) and \(\boldsymbol\beta\) represent the calibrated preference signals, refined through debiasing and denoising regularization. The analytical nature of the inference process makes exposure controllability an endogenous and interpretable outcome.
(i) The priors serve as \textit{policy-shape knobs}:  
\(\boldsymbol{\pi}^+\) specifies the desired exposure profile for the positive side (e.g., promoting diversity or category balance), shaping how \(\boldsymbol{\alpha}\) allocates mass across the \(M\) positive items.  
\(\boldsymbol{\pi}^-\) defines the intended suppression profile for the negative side (e.g., suppressing popular items or low-quality content), shaping how \(\boldsymbol{\beta}\) allocates mass over the \(N\) negative samples.  
Thus, the priors \((\boldsymbol{\pi}^+, \boldsymbol{\pi}^-)\) enable practitioners to instantiate task-specific exposure and suppression targets based on available signals, such as category, quality, novelty, and popularity. A simple example of prior encoding is provided in Eq.~\eqref{eq:prior1} and ~\eqref{eq:prior2}.
(ii) The temperature coefficients \(c_{\mathrm{pos}}\) and \(c_{\mathrm{neg}}\) act as \textit{strength knobs}, controlling the degree to which the posteriors adhere to the corresponding priors during inference. When \(c_{\mathrm{pos}}\) and \(c_{\mathrm{neg}}\) are small, inference is more influenced by preference-alignment signals, resulting in \textit{preference-driven inference}; when they are large, the priors exert greater influence, pulling the posteriors closer to the prescribed priors, resulting in \textit{policy-driven inference}.  
Consequently, VarBPR implements exposure control in an interpretable and analytically transparent manner: the priors specify \textit{what} pattern is desired, while the temperatures specify \textit{how strongly} the posterior should conform to it.

\subsection{Variational Learning Objective}
With the variational inference step completed, each enriched interaction
\(\mathbf x=(u,i_{1:M},j_{1:N})\) is associated with fixed mean-field variational
posteriors \(\boldsymbol\alpha\in\Delta^M\) and
\(\boldsymbol\beta\in\Delta^N\).
We now describe the variational learning stage (M-step), where we optimize the
ranking parameters \(\theta\) while keeping \((\boldsymbol\alpha,\boldsymbol\beta)\) fixed.
Since the priors \(\boldsymbol\pi^+,\boldsymbol\pi^-\)
are independent of \(\theta\), and \((\boldsymbol\alpha,\boldsymbol\beta)\) are fixed in this stage,
the regularization term in Eq.~\eqref{eq:elbo1} is constant w.r.t.\ \(\theta\)
and can be omitted.
Therefore, the \(\theta\)-dependent ELBO contribution for a single enriched interaction reduces to
\begin{equation}\label{eq:ideal-elbo2}
\mathcal J
=\sum_{m=1}^M\sum_{n=1}^N \alpha_m\beta_n\log \sigma\Big(\langle \mathbf{u}, \mathbf{i}^+_m\rangle - \langle \mathbf{u}, \mathbf{j}^-_n\rangle\Big)
\end{equation}
Evaluating \eqref{eq:ideal-elbo2} na\"ively scales as \(\mathcal{O}(MN)\) per interaction, which is undesirable
in large-scale recommendation settings. Next we show how to perform posterior compression via plug-in summarization.

\textbf{Posterior compression}. Let \((m,n)\sim(\boldsymbol\alpha,\boldsymbol\beta)\) and define the induced random margin
\(\Gamma \triangleq \gamma_{mn}(\theta)\), where
\(\gamma_{mn}(\theta)=\langle \mathbf{u}, \mathbf{i}^+_m\rangle - \langle \mathbf{u}, \mathbf{j}^-_n\rangle\).
Then the summarized margin admits a separable form:
\begin{equation}\label{eq:mu-separable}
\mathbb E\!\left[\Gamma\right]
= \langle \mathbf u,\mathbf c_u^+\rangle-\langle \mathbf u,\mathbf c_u^-\rangle,
\end{equation}
where the variational posterior-summarized centers are
\begin{equation}\label{eq:centers}
\mathbf c_u^+ \triangleq \sum_{m=1}^M \alpha_m \mathbf i_m^+,
\qquad
\mathbf c_u^- \triangleq \sum_{n=1}^N \beta_n \mathbf j_n^-.
\end{equation}
Importantly, Eq.~\eqref{eq:mu-separable} reduces the complexity of the summarized margin to \(\mathcal O(M+N)\). We next investigate the error bound arising from approximating the ELBO term $\mathbb{E}[\ell(\Gamma)]$ with the summarized margin $\ell\mathbb{E}(\Gamma)$, where $\ell(\gamma) \triangleq \log\sigma(\gamma)$.

\begin{proposition}[Controlled Jensen gap]\label{prop:jensen-gap}
Let  \(\Gamma=\gamma_{mn}(\theta)\) with \((m,n)\sim(\boldsymbol\alpha,\boldsymbol\beta)\).
If \(\mathrm{Var}(\Gamma)<\infty\), then
\begin{equation}\label{eq:jensen-gap}
\ell\!\big(\mathbb E[\Gamma]\big)-\tfrac18\,\mathrm{Var}(\Gamma)
\le
\mathbb E[\ell(\Gamma)]
\le
\ell\!\big(\mathbb E[\Gamma]\big).
\end{equation}
\end{proposition}
\noindent \textit{Proof sketch.}
Since \(\ell''(\gamma)=-\sigma(\gamma)\sigma(-\gamma)\in[-\tfrac14,0]\), \(\ell\) is concave, hence
\(\mathbb E[\ell(\Gamma)]\le \ell(\mathbb E[\Gamma])\) by Jensen's inequality.
For the lower bound, apply a second-order Taylor expansion of \(\ell(\Gamma)\) around \(\mathbb E[\Gamma]\):
\(\ell(\Gamma)=\ell(\mathbb E[\Gamma])+\ell'(\mathbb E[\Gamma])(\Gamma-\mathbb E[\Gamma])+\tfrac12 \ell''(\xi)(\Gamma-\mathbb E[\Gamma])^2\),
where \(\xi\) lies between \(\Gamma\) and \(\mathbb E[\Gamma]\). Taking expectations, the linear term vanishes and
\(\ell''(\xi)\ge -\tfrac14\) yields
\(\mathbb E[\ell(\Gamma)]\ge \ell(\mathbb E[\Gamma])-\tfrac18\mathbb E[(\Gamma-\mathbb E[\Gamma])^2]\),
which gives Eq.~\eqref{eq:jensen-gap}.

Proposition~\ref{prop:jensen-gap} shows that the plug-in summarization
\(\mathbb E[\ell(\Gamma)] \approx \ell(\mathbb E[\Gamma])\) incurs a second-order controlled error governed by
\(\mathrm{Var}(\Gamma)\).
Motivated by this bound, we introduce a linearly scaled objective as
\begin{equation}\label{eq:vbpr-per}
\tilde{\mathcal J}
\triangleq \ell(\mathbb E[\Gamma])=
\log\sigma\!\Big(\langle \mathbf u,\mathbf c_u^+\rangle-\langle \mathbf u,\mathbf c_u^-\rangle\Big),
\end{equation}

Eq.~\eqref{eq:vbpr-per} introduces a deployment-friendly plug-in approximation: it replace the expected log-likelihood \(\mathbb E[\ell(\gamma)]\) by its summarized form \(\ell(\mathbb E[\Gamma])\),
incurring a controlled second-order error governed by the posterior margin variance, while reduces the per-interaction computation from \(\mathcal O(MN)\) to \(\mathcal O(M+N)\).

Aggregating over the enriched training set \(\mathcal S\) and minimizing the negative average yields the
Variational BPR objective of the empirical form:
\begin{equation}\label{eq:vbpr-emp}
 \hat{\mathcal L}_{\mathrm{VarBPR}}(\theta)
=
-\frac{1}{|\mathcal S|}
\sum_{\mathbf x\in\mathcal S}
\log\sigma\!\Big(\langle \mathbf u,\mathbf c_u^+\rangle-\langle \mathbf u,\mathbf c_u^-\rangle\Big),
\end{equation}
where \((\boldsymbol\alpha,\boldsymbol\beta)\) are fixed during  learning stage, while \(\mathbf c_u^\pm\) still depend on \(\theta\)
through the encoder outputs in Eq.~\eqref{eq:centers}.

\subsection{Implementation}
VarBPR requires $M$ positive and $N$ negative items per user for variational inference. 
Following the BPR data format, we organize each training instance as $(u,i_{1:M},j_{1:N})$, which can be easily implemented by rewriting the \verb|collate_fn| in the \verb|DataLoader|. 
Algorithm~\ref{alg:algorithm} summarizes the training procedure for a single epoch. (i) extract user/item representations; (ii) solve the variational posteriors $\boldsymbol{\alpha,\beta}$ via Eqs.~\eqref{eq:var:alpha}--\eqref{eq:var:beta}; and (iii) compute the VarBPR loss using Eq.~\eqref{eq:vbpr-emp}.

\begin{algorithm}[!h]
	\caption{Pseudocode of VarBPR}
	\label{alg:algorithm}
	\textbf{Input}: Training set $\mathcal{S}$ with entries $\mathbf{x}=(u,i_{1:M},j_{1:N})$, encoder $f_\theta$, prior $(\boldsymbol{\pi}^+,\boldsymbol{\pi}^-)$, scaling factors $(c_{\mathrm{pos}},c_{\mathrm{neg}})$, learning rate $\eta$.
 \\
	\textbf{Output}: $\theta$. 
	\begin{algorithmic}[1] 
		\STATE Randomly initialize $\theta$.
        \FOR{each $\mathbf{x} \in \mathcal{S}$}
        \STATE \textbf{\textcolor{lightgray}{// Feature Extraction}}
		\STATE $\mathbf{u} \leftarrow f_\theta(u)$.
		\STATE $\mathbf i_{1},\cdots, \mathbf i_{M} \leftarrow f_\theta(i_{1}), \cdots, f_\theta(i_M)$.
		\STATE $\mathbf j_{1},\cdots, \mathbf j_{N} \leftarrow f_\theta(j_1), \cdots, f_\theta(j_N)$.
        \STATE \textbf{\textcolor{lightgray}{// Variational Inference}}
        \STATE Solve variational posteriors $\boldsymbol{\alpha,\beta}$ via Eqs.~\eqref{eq:var:alpha}-\eqref{eq:var:beta}.\label{algo:vi}
        \STATE \textbf{\textcolor{lightgray}{// Variational Learning}}
		\STATE Calculate VarBPR via Eq.~\eqref{eq:vbpr-emp}.\textcolor{lightgray}{\textbf{//} or Eq.\eqref{eq:ideal-elbo2}}\label{algo:loss}
		\STATE $\theta \leftarrow \theta - \eta\nabla_\theta  \hat{\mathcal L}_{\mathrm{VarBPR}}(\theta)$.
        \ENDFOR
    \STATE \textbf{Return} $\theta$.
	\end{algorithmic}
\end{algorithm}

\textit{Complexity Analysis}: For each data entry, VarBPR encodes one user along with \(M+N\) items, and computes scores for \(M\) positive and \(N\) negative samples. The variational inference (line~\ref{algo:vi}), loss calculation (line~\ref{algo:loss}), and even the prior construction for encoding hardness all reuse the scores obtained from these \(M+N\) items. Therefore, aside from feature extraction, the additional computational overhead is virtually negligible. In contrast, BPR only requires encoding one user and two items. Thus, the main overhead of VarBPR lies in the feature extraction step, which processes \(M+N\) items. Nevertheless, this additional computation remains strictly linear in time complexity, consistent with the computational overhead comparisons reported in the Section~\ref{sec:time}.

\section{Analysis of Variational BPR}
\subsection{Generalization Analysis}
Let \(\mathcal{X}\) be the space of interactions \(\mathbf x\), and suppose training samples \(\mathcal S=\{\mathbf{x}_t\}_{t=1}^n\) are drawn i.i.d.\ from some unknown distribution \(\mathcal{P}\) over \(\mathcal{X}\).  
Given $\theta$, the VarBPR margin on $\mathbf x$ is
\[
\gamma_{\text{v}}(\mathbf x;\theta)
=
\langle \mathbf u,\mathbf c_u^+(\mathbf x)\rangle
-
\langle \mathbf u,\mathbf c_u^-(\mathbf x)\rangle,
\]
where $\mathbf{u}$, $\mathbf{c}_u^{+}(\mathbf{x})=\sum\alpha_m(\mathbf{x})\,\mathbf{i}_m$ and $\mathbf{c}_u^{-}(\mathbf{x})=\sum\beta_n(\mathbf{x})\,\mathbf{j}_n$, are all induced by $f_\theta$, and the pairwise logistic loss is $\,\ell(\gamma) = -\log\sigma(\gamma)$. So $\ell'(\gamma)=\sigma(\gamma)-1\in[-1,0]$, hence $\ell$ is $1$-Lipschitz in $\gamma$. We define the population and empirical VarBPR risks as
\[
\mathcal L_{\text{v}}(\theta)
=
\mathbb E_{\mathbf x\sim\mathcal P}\big[\ell(\gamma_\text{v}(\mathbf x;\theta))\big],
\qquad
\widehat{\mathcal L}_{\text{v}}(\theta)
=
\frac1n \sum_{t=1}^n \ell(\gamma_\text{v}(\mathbf x_t;\theta)).
\]

Let the function class be
\[
\mathcal G
=
\big\{
  g_\theta(\mathbf x) := \ell(\gamma_{\mathrm v}(\mathbf x;\theta))
  \,:\, \theta\in\Theta
\big\}.\]
Without loss of generality, we assume that user and item embeddings are normalized.
Since $\mathbf{c}_u^{\pm}$ are convex combinations of unit-norm item embeddings, $\|\mathbf{c}_u^{\pm}\| \leq 1$, so that
$|\gamma(\mathbf x;\theta)| \le 2$ for all $\mathbf x$ and $\theta$, and hence the loss is uniformly bounded:
$0 \le g_\theta(\mathbf x) = \ell(\gamma_{\mathrm v}(\mathbf x;\theta))
\le B_\ell := \log(1+e^{2}),
\quad \forall\, g_\theta\in\mathcal G.$

Given training set $S = \{\mathbf x_1,\dots,\mathbf x_n\}$, its empirical Rademacher complexity is
\[
\widehat{\mathfrak R}_S(\mathcal G)
=
\mathbb E_{\boldsymbol\sigma}
\Big[
  \sup_{\theta \in\Theta}
  \frac1n \sum_{t=1}^n \sigma_t g_\theta (\mathbf x_t)
\Big],
\]
where $\sigma_t$ are i.i.d.\ Rademacher random variables taking values $\pm1$ with equal probability.
The population Rademacher complexity is
\[
\mathfrak R_n(\mathcal G)
=
\mathbb E_{S\sim(\mathcal P)^n}
\big[
  \widehat{\mathfrak R}_S(\mathcal G)
\big].
\]

We further introduce the \textit{clean-center risk} as an idealized objective, evaluated using an unobservable clean-center pair associated with each enriched interaction $\mathbf x$:
\[
\mathcal L^*(\theta)=\mathbb E_{\mathbf x\sim\mathcal P}\big[\ell(\gamma^*(\mathbf x;\theta))\big],
\]
where $\gamma^*(\mathbf x;\theta)
=\langle \mathbf u,\mathbf c_u^{*+}(\mathbf x)\rangle-\langle \mathbf u,\mathbf c_u^{*-}(\mathbf x)\rangle$ is the clean margin evaluated on clean centers:
\begin{align*}
\mathbf c_u^{*+}(\mathbf x)
=\sum\alpha_m^*(\mathbf x)\,\mathbf i_m,\qquad
\mathbf c_u^{*-}(\mathbf x)=\sum\beta_n^*(\mathbf x)\,\mathbf j_n.
\end{align*}
Here $\boldsymbol\alpha^*(\mathbf x) \in \Delta^M$ and $\boldsymbol\beta^*(\mathbf x)\in\Delta^N$ denote oracle variational posteriors over bag indices that  ideally place their probability mass on the clean samples. Because the variational family spans the entire simplex over bag indices, we can equivalently view $(\boldsymbol\alpha^*,\boldsymbol\beta^*)$ as the optima of the same variational subproblems in Eqs.~\eqref{eq:var-pos} and \eqref{eq:var-neg} under oracle variational-inference parameters $(c_\text{pos}^*,c_\text{neg}^*)$ and oracle priors $(\boldsymbol{\pi}^{*+},\boldsymbol{\pi}^{*-})$.

Variational inference can then be viewed as an indexing map that projects a noisy bag of samples to a surrogate center approximating this ideal target. Next, we establish a generalization bound that relates the empirical VarBPR risk to this ideal \textit{clean-center} objective.

\vspace{5pt}
\begin{theorem}[Generalization bound]\label{thm:generalization}
Let \(\mathcal{G}\) be the function class induced by the encoder \(f_\theta\) and the scoring margin \(\gamma\).  For any \(\delta \in (0,1)\), with probability at least \(1-\delta\),
\begin{align}
\mathcal{L}^*(\theta) - \widehat{\mathcal{L}}_\text{v}(\theta)
\le
\mathcal C_\mathsf{KL}+2\mathfrak R_n(\mathcal G)+B_\ell \sqrt{\frac{\ln(1/\delta)}{2n}},
\end{align}
where \(\mathcal{C}_{\mathsf{KL}} = c_{\text{pos}}^{*} \mathbb{E}_\mathbf x \mathrm{KL}\left(\boldsymbol{\alpha} \Vert \boldsymbol{\pi}^{*+}\right)
          + c_{\text{neg}}^{*} \mathbb{E}_\mathbf x \mathrm{KL}\left(\boldsymbol{\beta} \Vert \boldsymbol{\pi}^{*-}\right)\) 
          is the expected posterior--oracle-prior KL penalty with fixed constants $c_{\mathrm{pos}}^*,c_{\mathrm{neg}}^*$.
\end{theorem}

\begin{proof}
We decompose the gap between the ideal risk and the empirical VarBPR risk as
\begin{align}
\mathcal L^*(\theta) - \widehat{\mathcal L}_{\mathrm{v}}(\theta)
=\underbrace{\mathcal L^*(\theta)- \mathcal L_{\mathrm{v}}(\theta)}_{\text{indexing error}}
+
\underbrace{ \mathcal L_{\mathrm{v}}(\theta) - \widehat{\mathcal L}_{\mathrm{v}}(\theta)}_{\text{generalization error}}.\label{eq:error-decomposition}
\end{align}

\textbf{(i) Indexing error.} 
We first bound the structural mismatch between the clean-center objective and the VarBPR objective through a concise variational argument. On the positive side, consider an oracle poster $\boldsymbol\alpha^*$ over bag indices that ideally place their probability mass on clean samples. Because the variational family spans the full simplex, we can equivalently view $\boldsymbol\alpha^*$ as the optima of an oracle variational problem in the same form as Eq.~\eqref{eq:var-pos}, but parameterized by oracle variational-inference parameter $c_\text{pos}^*$ and oracle prior $\boldsymbol{\pi}^{*+}$, for which the oracle variational posterior
$\boldsymbol{\alpha}^*$ is the unique maximizer. Plugging the VarBPR variational posterior $\boldsymbol{\alpha}$ as a feasible candidate into this oracle variational problem gives
\[
\mathcal{V}^{+}(\boldsymbol{\alpha}^*)
\;\ge\;
\mathcal{V}^{+}(\boldsymbol{\alpha}).
\]
Substituting $\sum_m\alpha_m\langle\mathbf{u},\mathbf{i}_m \rangle = \langle\mathbf{u},\mathbf{c}_u^+ \rangle$ and expanding the regularizer as a KL divergence yields
\begin{equation}
\langle \mathbf u,\mathbf c_u^{*+}\rangle-\langle \mathbf u,\mathbf c_u^{+}\rangle
\;\ge\;
c_{\text{pos}}^*\big(
\mathrm{KL}(\boldsymbol\alpha^*\Vert\boldsymbol\pi^{*+})
-
\mathrm{KL}(\boldsymbol\alpha\Vert\boldsymbol\pi^{*+})
\big).
\label{eq:pos-gap-clean}
\end{equation}

Analogously, on the negative side we consider an oracle variational subproblem
of the same form as Eq.~\eqref{eq:var-neg}, parameterized by oracle variational-inference parameters $c_\text{neg}^*$ and oracle prior $\boldsymbol{\pi}^{*-}$. Using the VarBPR posterior
$\boldsymbol{\beta}$ as a feasible candidate yields
\begin{equation}
\langle \mathbf u,\mathbf c_u^{*-}\rangle-\langle \mathbf u,\mathbf c_u^{-}\rangle
\;\le\;
c_{\text{neg}}^*\big(
\mathrm{KL}(\boldsymbol\beta\Vert\boldsymbol\pi^{*-})
-
\mathrm{KL}(\boldsymbol\beta^*\Vert\boldsymbol\pi^{*-})
\big).
\label{eq:neg-gap-clean}
\end{equation}

Combining \eqref{eq:pos-gap-clean}–\eqref{eq:neg-gap-clean} gives a lower bound
on the margin difference:
\begin{small}
\begin{align}
&\gamma_{\text{v}}(\mathbf x)-\gamma^*(\mathbf x) =\langle \mathbf u,\mathbf c_u^+\rangle
-\langle \mathbf u,\mathbf c_u^-\rangle - 
\langle \mathbf u,\mathbf c_u^{*+}\rangle
+\langle \mathbf u,\mathbf c_u^{*-}\rangle\nonumber\\
\le& 
-\,c_{\text{pos}}^*\Big(
\mathrm{KL}(\boldsymbol\alpha^*\Vert\boldsymbol\pi^{*+})
-
\mathrm{KL}(\boldsymbol\alpha\Vert\boldsymbol\pi^{*+})
\Big)\nonumber\\
&+
c_{\text{neg}}^*\Big(
\mathrm{KL}(\boldsymbol\beta\Vert\boldsymbol\pi^{*-})
-
\mathrm{KL}(\boldsymbol\beta^*\Vert\boldsymbol\pi^{*-})
\Big).
\end{align}  
\end{small}

\par Since the oracle hyperparameters are fixed and satisfy $c_{\mathrm{pos}}^*>0$, $c_{\mathrm{neg}}^*>0$, and $\mathrm{KL}(\cdot\Vert\cdot)\ge 0$, dropping the non-positive terms yields
\begin{align}\label{eq:margin-gap}
    \big(\gamma_{\text{v}}(\mathbf x)-\gamma^*(\mathbf x)\big)_+\leq c_{\text{pos}}^*\mathrm{KL}(\boldsymbol\alpha\Vert\boldsymbol\pi^{*+}) + c_{\text{neg}}^*\mathrm{KL}(\boldsymbol\beta\Vert\boldsymbol\pi^{*-}),
\end{align} 
where $(x)_+=\max\{x,0\}$ denotes the positive part. Recall that $\ell(\gamma)=-\log\sigma(\gamma)$ is monotonically decreasing and $1$-Lipschitz in $\gamma$.
Thus,
\begin{equation}
\ell(\gamma^*(\mathbf x))-\ell(\gamma_{\text{v}}(\mathbf x))
\;\le\;
\big(\gamma_{\text{v}}(\mathbf x)-\gamma^*(\mathbf x)\big)_+.
\label{eq:lipschitz-single-side}
\end{equation}
Taking expectation over $\mathbf x$ yields
\begin{equation}
\mathcal L^*-\mathcal L_{\text{v}}
\;\le\;
\mathbb E_\mathbf x\big[
  \big(\gamma_{\text{v}}(\mathbf x)-\gamma^*(\mathbf x)\big)_+
\big].
\label{eq:index-loss-gap}
\end{equation}
Substituting Eq.~\eqref{eq:margin-gap} into Eq.~\eqref{eq:index-loss-gap} yields 
\begin{equation}
\mathcal L^*-\mathcal L_{\text{v}}
\;\le\;
c_{\text{pos}}^*\mathbb E_\mathbf x\mathrm{KL}(\boldsymbol\alpha\Vert\boldsymbol\pi^{*+}) + c_{\text{neg}}^*\mathbb E_\mathbf x\mathrm{KL}(\boldsymbol\beta\Vert\boldsymbol\pi^{*-}).
\label{eq:index-bound-main}
\end{equation}

\textbf{(ii) Generalization error.} We next employ the Rademacher complexity to derive the high-probability bound for the generalization error.
Recall that $0\le g_\theta(\cdot)\le B_\ell$. Define
\begin{align}
\Phi(\mathcal S)&=\sup_{\theta \in  \Theta}\big(\mathcal L_{\mathrm{v}}(\theta)
-\widehat{\mathcal L}_{\mathrm{v}}(\theta)
\big)\nonumber \\
&=
\sup_{\theta \in \Theta}\big(\mathbb E_\mathbf x[g_\theta(\mathbf x)]-
\frac1n\sum_{t=1}^n g_\theta(\mathbf x_t)\big).
\end{align}
Let $\mathcal S = \{\mathbf x_1,\dots,\mathbf x_t,\dots,\mathbf x_n\}$ and $\tilde {\mathcal S} = \{\mathbf x_1,\dots,\tilde{\mathbf x}_t,\dots,\mathbf x_n\}$ differ only at index $t$. Noting that $\mathcal L_\text{v}(\theta)$ does not depend on $\mathcal{S}$, and using 
$
\left| \sup_{\theta} a_{\theta} - \sup_{\theta} b_{\theta} \right| \le \sup_{\theta} \left| a_{\theta} - b_{\theta} \right|
$, we have
\begin{align}
|\Phi(\mathcal S) - \Phi(\tilde {\mathcal S})|&\le\sup_{\theta \in  \Theta}\big|\widehat{\mathcal L}_{\mathrm{v}}(\theta;\mathcal S)-\widehat{\mathcal L}_{\mathrm{v}}(\theta;\tilde {\mathcal S})\big|\nonumber \\
&= \sup_{\theta \in  \Theta}\big|  \frac1n\big(g_\theta(\mathbf x_t) - g_\theta(\tilde{\mathbf x}_t)\big) \big|\nonumber \\
&\leq B_\ell/n.  
\end{align}
Thus $\Phi(\mathcal S)$ satisfies a bounded differences condition with constants
$c_t = B_\ell/n$ for $t=1,\dots,n$. By McDiarmid's inequality in the one-sided form
\[
\mathbb P\big(\Phi(\mathcal S) - \mathbb E[\Phi(\mathcal S)] \ge \varepsilon\big)
\;\le\;
\exp\Big(
  -\frac{2\varepsilon^2}{\sum_{t=1}^n c_t^2}
\Big),
\]
and using $\sum_{t=1}^n c_t^2 = n\cdot (B_\ell/n)^2 = B_\ell^2/n$, we obtain
\[
\mathbb P\big(\Phi(\mathcal S) - \mathbb E[\Phi(\mathcal S)] \ge \varepsilon\big)
\;\le\;
\exp\Big(
  -\frac{2n\varepsilon^2}{B_\ell^2}
\Big).
\]
Setting the right-hand side equal to $\delta$ and solving for $\varepsilon$ yields $\varepsilon=B_\ell \sqrt{\frac{\ln(1/\delta)}{2n}}$. Therefore, with probability at least $1-\delta$,
\begin{equation}
\Phi(\mathcal S)\;\le\;
\mathbb E[\Phi(\mathcal S)]
+
B_\ell \sqrt{\frac{\ln(1/\delta)}{2n}}.\nonumber
\end{equation}

We next introduce an independent "ghost sample" $\mathcal{S}' = \{\mathbf{x}'_1,\dots,\mathbf{x}'_n\}$ with $\mathbf{x}'_t\sim\mathcal{P}$ i.i.d., along with i.i.d.\ Rademacher variables $\sigma_1,\dots,\sigma_n$. We then derive a bound on $\mathbb{E}[\Phi(\mathcal{S})]$ in terms of Rademacher complexity using the standard symmetrization procedure.
\begin{small}
\begin{align*}
\mathbb E_S[\Phi(\mathcal S)]
&=\mathbb E_\mathcal S \sup_{\theta \in  \Theta }\Big(\mathbb E_{\mathbf x}[g_\theta(\mathbf x)]
-
\frac1n\sum_{t=1}^n g_\theta(\mathbf x_t)\Big)\\
&=\mathbb E_\mathcal S
\sup_{\theta \in  \Theta}
\Big(\mathbb E_{\mathcal S'}\frac1n\sum_{t=1}^n g_\theta(\mathbf x'_t)-\frac1n\sum_{t=1}^n g_\theta(\mathbf x_t)\Big) \\
&\le
\mathbb E_{\mathcal S,\mathcal S'}
\sup_{\theta \in  \Theta}
\Big(\frac1n\sum_{t=1}^n g_\theta(\mathbf x'_t)-\frac1n\sum_{t=1}^n g_\theta(\mathbf x_t)\Big)\\
&=\mathbb E_{\mathcal S,\mathcal S',\boldsymbol\sigma}\sup_{\theta \in  \Theta}\frac1n\sum_{t=1}^n \sigma_t \big(g_\theta(\mathbf x'_t)-g_\theta(\mathbf x_t)\big)\\
&\le
\mathbb E_{\mathcal S',\boldsymbol\sigma}\sup_{\theta \in  \Theta}\frac1n\sum_{t=1}^n \sigma_t g_\theta(\mathbf x'_t)
+\mathbb E_{\mathcal S,\boldsymbol\sigma}\sup_{\theta \in \Theta}\frac1n\sum_{t=1}^n \sigma_t g_\theta(\mathbf x_t)\\
&=2\mathfrak R_n(\mathcal G).
\end{align*}
\end{small}
So, with probability at least $1-\delta$,
\begin{equation}\label{eq:Phi-high-prob}
\mathcal L_{\mathrm{v}}(\theta)
-\widehat{\mathcal L}_{\mathrm{v}}(\theta)\leq\Phi(\mathcal S)\le2\mathfrak R_n(\mathcal G)+B_\ell \sqrt{\frac{\ln(1/\delta)}{2n}}.
\end{equation}
Since the indexing error bound in Eq.~\eqref{eq:index-bound-main} holds deterministically for any $\mathcal S$, intersecting it with the $1-\delta$ high-probability event in Eq.~\eqref{eq:Phi-high-prob} yields the claimed bound with probability at least $1-\delta$.

\end{proof}

\textbf{Generalization insight}. Our bound decomposes the gap to the clean-center risk of VarBPR into two parts:
(i) an indexing/structural error term that captures the structural bias induced by the variational inference, and
(ii) a statistical generalization term governed by the bag-level Rademacher complexity of the function class $\mathcal G$ together with a $\tilde{\mathcal O}(1/\sqrt n)$ concentration term.
Thus, with sufficiently many samples and an approximately optimized VarBPR objective, the clean-center risk approaches the best achievable level within the hypothesis class up to a residual offset governed solely by the posterior--oracle-prior KL complexity term.

We next examine the KL divergence term between the posterior and the oracle prior. Focusing on the positive side (with the negative side term $\mathrm{KL}\left(\boldsymbol{\beta} \Vert \boldsymbol{\pi}^{*-}\right)$ following analogously), it admits the decomposition
\begin{small}
\begin{equation}
\mathrm{KL}(\boldsymbol\alpha\Vert\boldsymbol\pi^{*+})
=
\underbrace{\mathrm{KL}(\boldsymbol\alpha\Vert\boldsymbol\pi^+)}_{\text{variational misspecification}}
+
\underbrace{H(\boldsymbol\alpha,\boldsymbol\pi^{*+}) - H(\boldsymbol\alpha,\boldsymbol\pi^+)}_{\text{prior misspecification}}.
\end{equation}
\end{small}

This decomposition makes the KL complexity interpretable as two distinct failure modes.
The first term is a \textit{compliance gap}: it quantifies how much the learned posterior still ``overrides'' the prescribed exposure policy (e.g., long-tail oriented). Increasing the regularization strengths $c_{\mathrm{pos}},c_{\mathrm{neg}}$ (and the bag sizes $M,N$) progressively suppresses this gap by pulling $\boldsymbol\alpha$ toward $\boldsymbol\pi^+$, i.e., shrinking the effective degrees of freedom of variational posterior. The second term, $H(\boldsymbol\alpha,\boldsymbol\pi^{*+})-H(\boldsymbol\alpha,\boldsymbol\pi^+)=\mathbb E_{\boldsymbol\alpha}\!\big[\log \tfrac{\pi^+}{\pi^{*}}\big]$, is a \textit{policy-quality gap}: it measures the cost of enforcing an imperfect prior on the mass locations where $\boldsymbol\alpha$ concentrates.
In the high-compliance regime where $\boldsymbol\alpha\simeq\boldsymbol\pi^+$, it is well approximated by $\mathrm{KL}(\boldsymbol\pi^+\Vert\boldsymbol\pi^{*})$,
reflecting an irreducible ``\textbf{opportunity cost}'' incurred by using a hand-specified (potentially imperfect) $\boldsymbol\pi^+$ instead of the oracle $\boldsymbol\pi^{*+}$.
Consequently, overly strong regularization no longer learns from data, but instead locks in the residual determined by prior quality. Thus, the regularization strength induces a trade-off between compromising toward the prior (e.g., for fairness/debiasing exposure control) and the opportunity cost due to prior mismatch. Excessively strong regularization (implementing stronger exposure control) may reduce the first term, but it also makes the irreducible opportunity cost stemming from prior bias the dominant source of error.

Practically, this yields a clear \textbf{tuning principle}: when the designed prior $\boldsymbol\pi$ is trusted (e.g., a reliable long-tail or fairness policy), one can safely strengthen regularization so that the remaining indexing error is dominated by a small opportunity-cost floor; when $\boldsymbol\pi$ is uncertain, moderate regularization is preferable to avoid locking in a potentially large ``opportunity cost'' incurred by prior mismatch .

\subsection{Counteracting Popularity Bias}
In the variational inference step of VarBPR, the core process involves inferring the true intent signal—represented by the variational posterior \((\boldsymbol{\alpha}, \boldsymbol{\beta})\)—from a bag of noisy samples, guided by a regularized preference alignment principle that incorporates both debiasing and denoising. This inference mechanism introduces two directly controllable factors that enable the precise regulation of long-tail item exposure: the prior distribution shape and the regularization strength.

(i) \textit{Prior distribution}: By defining exposure profiles \(\boldsymbol{\pi}^+\) and suppression profiles \(\boldsymbol{\pi}^-\), the prior directly steers the inferred variational posteriors \((\boldsymbol{\alpha}, \boldsymbol{\beta})\) towards the desired exposure and suppression patterns, thereby shaping both the true positive and negative interest distributions. Specifically, in positive interest modeling, \(\boldsymbol{\pi}^+\) prevents \(\boldsymbol{\alpha}\) from concentrating too heavily on popular positives, thereby increasing the representation of niche, long-tail items. In negative interest modeling, \(\boldsymbol{\pi}^-\) allocates more weight to popular negative items, enhancing their suppression and reducing their interference with the true preference structure.

(ii) \textit{Regularization strength}: The temperature parameters \(c_{\mathrm{pos}}\) and \(c_{\mathrm{neg}}\) control the extent to which \((\boldsymbol{\alpha}, \boldsymbol{\beta})\) adhere to these prior distributions. This mechanism provides a smooth trade-off between preference-driven inference and policy-driven inference, ultimately enabling controlled long-tail exposure.

\subsection{Noise Mitigation}
A straightforward algebraic manipulation of the VarBPR objective then yields:
\[
\mathcal{L} = \frac{1}{|\mathcal{S}|} \sum \ln \frac{\exp\langle\mathbf{u}, \mathbf{c}_u^+\rangle}{\exp\langle\mathbf{u}, \mathbf{c}_u^+\rangle + \exp\langle\mathbf{u}, \mathbf{c}_u^-\rangle}.
\]
Mathematically, this formulation is equivalent to a contrastive learning objective with the user embedding \(\mathbf{u}\) as the anchor, pulling the positive interest center \(\mathbf{c}_u^+\) closer and pushing the negative center \(\mathbf{c}_u^-\) away. In contrast to instance-level contrastive learning that treats all samples equally, Variational BPR employs entropy-regularized variational objectives in Eq.~\eqref{eq:var-pos} and Eq.~\eqref{eq:var-neg} to regulate sample influence. Thus, as illustrated in Fig.~\ref{fig:deno}, each sample is assigned an  weight smaller than 1,  thereby reducing its potential disruptive effect on optimization. 
Viewed through the lens of self‑supervised learning, VarBPR evolves the learning signal from a rigid binary rule (“any interacted item is preferred”) to a smooth hypothesis: “the inferred positive center is preferred over a inferred negative center.” Thus, the model learns to differentiate genuine preference signals from noise, resulting in a more refined and robust characterization of user preferences.
\begin{figure}[h!]
	\vspace{-5pt}
	\centering
	\includegraphics[width=0.49\textwidth]{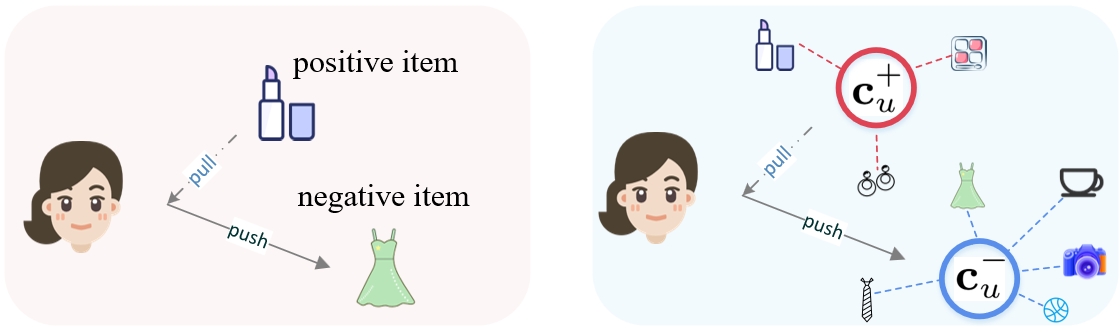}
	\caption{Noise mitigation effect. Left: Instance-level contrast.
Right: Variational attention-based prototype contrast, entropy regularization constrains the impact of noisy samples by limiting their weights smaller than 1.}
	\label{fig:deno}
\end{figure}

\subsection{Mining Hard Samples} 
By sorting candidate items by their similarity to the user embedding, we expose the ambiguous region near the decision boundary (white area in Fig.~\ref{fig:hard}(a)), where informative hard samples reside. Concretely, for a set of $M$ sampled positive candidates $\{i_m\}_{m=1}^M$ ordered by $\langle \mathbf u,\mathbf i_m\rangle$ in ascending order, $i_1$ corresponds to the hardest positive and $i_M$ to the easiest one. Recall that VarBPR infers posterior intent signals from a bag of samples. In its variational inference procedure, this naturally incorporates two complementary mechanisms for mining hard samples---on both the positive and negative sides:
(i) The scaling factors \(c_{\mathrm{pos}}\) and \(c_{\mathrm{neg}}\) explicitly control the posterior mass assigned to the informative samples (Fig.~\ref{fig:hard}(b)); and
(ii) Sample hardness can be directly encoded in the priors \(\boldsymbol{\pi}^+\) and \(\boldsymbol{\pi}^-\) (for a concrete instantiation, see Eqs.~\eqref{eq:prior1}--\eqref{eq:prior2}), thereby guiding the variational posteriors toward hard samples in a principled way.

\begin{figure}[h!]
	\vspace{-5pt}
	\centering
	\includegraphics[width=0.49\textwidth]{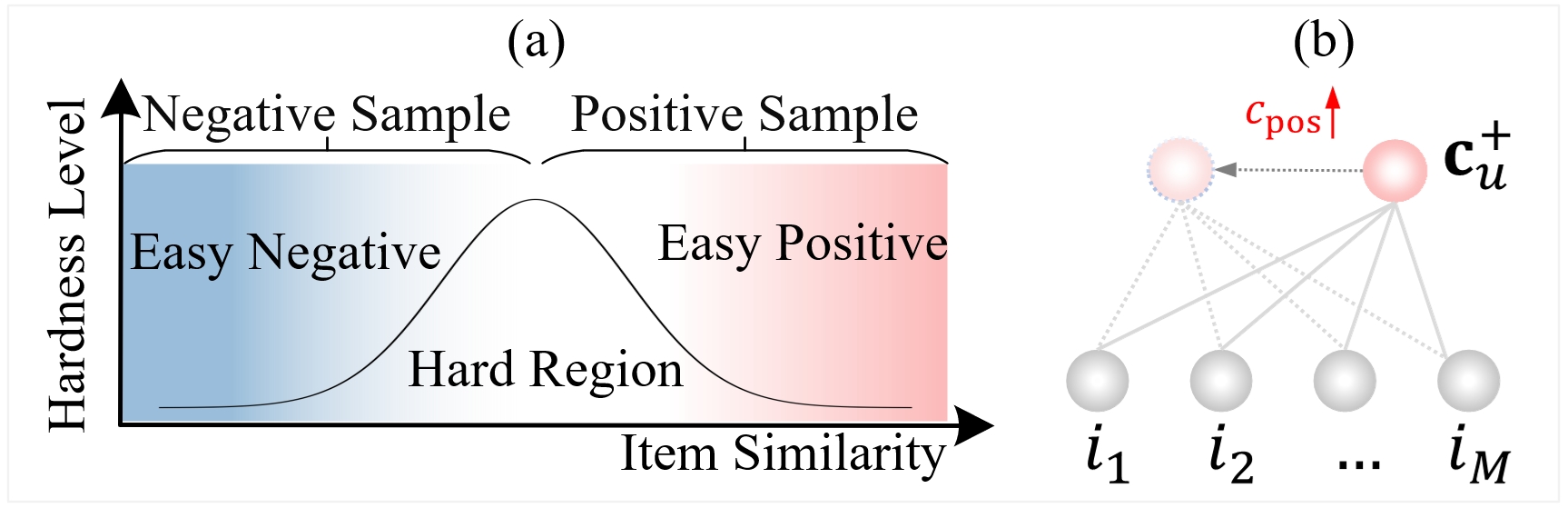}
	\caption{Variational inference enables controllable hard sample mining for both positive and negative sample.}
	\label{fig:hard}
\end{figure}

\subsection{Relationship to Previous Work}
VarBPR builds upon yet meaningfully departs from several related lines of work. Unlike prototype-based contrastive methods (e.g., NCL~\cite{lin2022improving}) that suppress noise through clustering, VarBPR leverages the entropy regularization inherent in variational inference for noise reduction. Compared to virtual sample synthesis method like MixGCF~\cite{10.1145/3447548.3467408}, the interest centers in VarBPR result from a plug-in approximation, which serves as a practical trade-off between performance and computational efficiency, with the ultimate aim of approximately optimizing the evidence lower bound (ELBO). Finally, VarBPR shares similarities with the Expectation-Maximization (EM) algorithm~\cite{Dempster:1977:RSS}, where the variational inference step for $\boldsymbol{\alpha}$ and $\boldsymbol{\beta}$ corresponds to the E-step, and the gradient update in VarBPR corresponds to the M-step.

\section{Experiments}
In the experimental section, we address the following research questions.
\textbf{RQ1 (Effectiveness):} Does VarBPR outperform strong baselines in top-$K$ ranking accuracy?
\textbf{RQ2 (Exposure controllability):} Can VarBPR provide controllable long-tail exposure via its analytical variational inference mechanism?
\textbf{RQ3 (Scalability):} How do the training costs of VarBPR compare with standard BPR?
\textbf{RQ4 (Ablation):} What is the contribution of each component---priors, variational posteriors $(\boldsymbol{\alpha},\boldsymbol{\beta})$, and plug-in compression?
\textbf{RQ5 (Compression validity):} Is the plug-in approximation error (Jensen gap) consistent with the theoretical bound?
\textbf{RQ6 (Robustness):} Is VarBPR more robust to noisy interactions?

\subsection{Experimental Setup}

\subsubsection{Datasets} To ensure a comprehensive evaluation, we conduct experiments on four widely-used public benchmarks from the recommendation literature, covering both explicit ratings and implicit feedback scenarios.
Specifically, we include two rating datasets: MovieLens-100K and MovieLens-1M~\cite{harper2015movielens}.
In addition, we adopt two implicit feedback datasets: Yelp 2018~\cite{Xiangnan:2020:SIGIR}, and Gowalla~\cite{Wang:2019:SIGIR}.
These datasets vary in domain, scale, and interaction sparsity, allowing us to assess the robustness and generalizability of the proposed approach across different recommendation settings.

\subsubsection{Evaluation protocols} For the MovieLens100k and MovieLens1M datasets, we randomly selected 50\% of the items rated 4 or higher by each user as a “clean" test set~\cite{wang2021denoising}. The remaining data was converted into implicit feedback for model training. The training set included false-positive examples (low-rated items) and false-negative examples (items the user might like but has not seen). For the Yelp2018 and Gowalla datasets, we randomly selected 20\% of interactions for testing, with the remaining 80\% used for training. Similar to the rating datasets, the training set included false-positive and false-negative examples. Although we couldn't create a perfectly clean test set from implicit feedback, the noise level is low enough to reasonably assess personalized recommendation accuracy, which is a common practice~\cite{10.1145/3289600.3291027,Xiangnan:2020:SIGIR,wan2022cross,10158930}. In line with existing studies~\cite{Wang:2019:SIGIR, Xiangnan:2020:SIGIR}, we evaluated ranking accuracy using Recall@20 and NDCG@20. Additionally, we assessed the \textbf{A}verage \textbf{P}roportional \textbf{L}ong-\textbf{t}ail items in top-k recommendations using  APLT@20, where long-tail items are defined as the bottom 85\% of items in the item popularity distribution. 

\subsubsection{Baseline methods}
We systematically evaluate Variational BPR (VarBPR) through a carefully designed set of comparative experiments. The selected baselines cover major technical lines in ranking optimization, which we organize into six categories to ensure both breadth and depth in the comparison:
(i) Foundational Ranking Losses: including the classic Bayesian Personalized Ranking (BPR)~\cite{Steffen:2009:UAI} and the contrastive-learning-based InfoNCE (also called Sampled Softmax)~\cite{Oord:2018:arxiv,Jiancan:2022:arxiv}.
(ii) Improved Pairwise Losses: covering group-behavior-aware GBPR~\cite{Weike:2013:IJCAI} and cross-domain contrastive CPR~\cite{wan2022cross}.
(iii) Noise-Reduction Methods: including denoising-oriented ADT~\cite{wang2021denoising} and self-guiding SGDL~\cite{gao2022self}.
(iv) Statistical Debiasing Techniques: involving causality-disentangling DCL~\cite{Chuang:2020:NIPS}, unbiased UBPR~\cite{saito2020unbiased}, and practical-debiasing PUPL~\cite{cao2024practically}.
(v) Popularity Debiasing Approaches: including popularity-disentangled PopDCL~\cite{Liupopdcl} and bias-combining BC Loss~\cite{zhang2022incorporating}.
(vi) Hard Negative Mining Strategies: represented by hard contrastive loss HCL~\cite{Robinson:2021:ICLR} and adversarially-enhanced AdvInfoNCE~\cite{zhang2024empowering}.
This multi-faceted comparison framework is designed to clearly delineate the advantages and limitations of VarBPR relative to existing paradigms.

\begin{table*}[!ht]
	\caption{The performance comparison is conducted using the Matrix Factorization (MF) backbone. The best results are highlighted in bold, while the second-best results are underlined. The improvement of VarBPR is calculated relative to the second-best results. The improvement achieved by VarBPR is statistically significant ($p$-value $<$ 0.01).}\label{tab:exp1}
	\centering
	\resizebox{0.99\textwidth}{!}{
		\begin{tabular}{llc>{\columncolor{gray!15}}c>{\columncolor{gray!15}}ccc>{\columncolor{gray!15}}c>{\columncolor{gray!15}}c}
			\toprule
			\multirow{2}{*}{Method} &   \multicolumn{2}{c}{MovieLens100K} &\multicolumn{2}{c}{\cellcolor{gray!15}MovieLens1M} &\multicolumn{2}{c}{Gowalla}&\multicolumn{2}{c}{\cellcolor{gray!15}Yelp2018} \\\cmidrule(lr){2-9}
			& Recall@20 & NDCG@20 & Recall@20 & NDCG@20 & Recall@20 & NDCG@20 & Recall@20 & NDCG@20 \\\midrule
			BPR~\cite{Steffen:2009:UAI} & 0.3226 & 0.4374 & 0.2175 & 0.4153 &  0.1408 & 0.1063 & 0.0704 & 0.0584  \\
			GBPR~\cite{Weike:2013:IJCAI} & 0.3049 & 0.4046 & 0.2195 & 0.4157 & 0.1543 & 0.1185 & 0.0723 & 0.0593 \\
			InfoNCE~\cite{Oord:2018:arxiv} & 0.3265 & 0.4436 & 0.2210 & 0.4214 & 0.1691 & 0.1273 & 0.0876 &0.0723 \\
			UBPR~\cite{saito2020unbiased} & 0.3253 & 0.4392 & 0.2168 & 0.4172 &  0.1411 & 0.1082 & 0.0719 & 0.0592  \\
			DCL~\cite{Chuang:2020:NIPS} & 0.3265 & 0.4446 & 0.2216 & 0.4228 & 0.1712 & 0.1293 & 0.0881  &0.0725  \\
			HCL~\cite{Robinson:2021:ICLR}        & \underline{0.3270} &\underline{0.4482} & 0.2225 & 0.4239 & \underline{0.1754} & 0.1318 & 0.0885 & 0.0728 \\
            ADT~\cite{wang2021denoising} & 0.3260 & 0.4432 & 0.2209 & 0.4209 & 0.1688 & 0.1270 & 0.0871 &0.0726 \\
            SGLD~\cite{gao2022self} & 0.3264 & 0.4442 & 0.2230 & 0.4212 & 0.1721 & 0.1302 & 0.0880 & 0.0721 \\
			BC Loss~\cite{zhang2022incorporating}& 0.3268 & 0.4452 & 0.2120 & 0.4239 & 0.1724 & 0.1290 & 0.0882 & 0.0726 \\
			CPR~\cite{wan2022cross}  & 0.3263 & 0.4408 & 0.2190 & 0.4176 & 0.1421 & 0.1075 & 0.0717 & 0.0596 \\
			PopDCL~\cite{Liupopdcl} & 0.3264 & 0.4458 & 0.2235 & 0.4223 & 0.1726 & 0.1308 & 0.0883 & 0.0724 \\
			AdvInfoNCE~\cite{zhang2024empowering} & 0.3260 & 0.4473 & 0.2232 & \underline{0.4241} & 0.1750 & \underline{0.1325} & \underline{0.0889} & \underline{0.0732} \\
            PUPL~\cite{cao2024practically} & 0.3265 & 0.4468 & \underline{0.2233} & 0.4233 & 0.1739 & 0.1315 & 0.0888 & 0.0731 \\
			VarBPR & \textbf{0.3566} & \textbf{0.4919} &  \textbf{0.2396}& \textbf{0.4586} & \textbf{0.2032} & \textbf{0.1482} & \textbf{0.0991} & \textbf{0.0822} \\
			Improv.  & \textcolor{myred}{+9.05\%} &\textcolor{myred}{+9.75\%} &\textcolor{myred}{+ 7.30\%} &\textcolor{myred}{+8.13\%} &\textcolor{myred}{+15.85\%} &\textcolor{myred}{+11.85\%} &\textcolor{myred}{+11.47\% } &\textcolor{myred}{+ 12.30 \%} \\
			\bottomrule
		\end{tabular}
	}
\end{table*}

\begin{table*}[!h]
	\caption{Performance w.r.t. different backbones. The improvement is statistically significant ($p$-value $<$ 0.01).}
	\label{tab:exp2}
	\centering
	\resizebox{0.99\textwidth}{!}{
		\begin{tabular}{ll>{\columncolor{gray!15}}c>{\columncolor{gray!15}}ccc>{\columncolor{gray!15}}c>{\columncolor{gray!15}}c}
			\toprule
			\multirow{2}{*}{Dataset} &\multirow{2}{*}{LossFunc} &   \multicolumn{2}{c}{\cellcolor{gray!15}MF} &   \multicolumn{2}{c}{LightGCN}&   \multicolumn{2}{c}{\cellcolor{gray!15}XSimGCL}\\
			& & Recall@20 & NDCG@20& Recall@20 & NDCG@20& Recall@20 & NDCG@20\\ \midrule
			\multirow{3}{*}{MovieLens100K}&BPR & 0.3226& 0.4374& 0.3217& 0.4469& 0.3290& 0.4489 \\
			& VarBPR & 0.3566 & 0.4919& 0.3649& 0.5005& 0.3641& 0.5001\\
			&Improv. &\textcolor{myred}{ +10.53\%}&\textcolor{myred}{ +12.46\%}& \textcolor{myred}{+13.43\%}& \textcolor{myred}{ +11.99\%}& \textcolor{myred}{+10.66\%}& \textcolor{myred}{+11.40\%}\\\cmidrule(lr){1-2}
			
			\multirow{3}{*}{MovieLens1M}&BPR &  0.2175& 0.4153& 0.2207& 0.4369 & 0.2232& 0.4395\\
			& VarBPR &  0.2396& 0.4586 & 0.2428& 0.4704& 0.2445& 0.4815\\
			&Improv. & \textcolor{myred}{ +10.16\%}& \textcolor{myred}{ +10.43\%}&  \textcolor{myred}{ +10.01\%}&  \textcolor{myred}{ +9.56\%}&  \textcolor{myred}{+9.54\%}&  \textcolor{myred}{+6.51\%}\\\cmidrule(lr){1-2}
			
			\multirow{3}{*}{Gowalla} & BPR &  0.1408& 0.1063& 0.1756& 0.1314& 0.1791& 0.1365\\
			& VarBPR &  0.2032 & 0.1482& 0.2123& 0.1567& 0.2152& 0.1618\\
			&Improv. &\textcolor{myred}{ +44.32\%}&\textcolor{myred}{ +39.41\%}& \textcolor{myred}{ +20.90\%}& \textcolor{myred}{ +19.25\%}& \textcolor{myred}{ +20.16\%}& \textcolor{myred}{+18.53\%}\\\cmidrule(lr){1-2}
			
			\multirow{3}{*}{Yelp2018}&BPR &  0.0704& 0.0584 & 0.0818& 0.0677& 0.0902& 0.0741\\
			& VarBPR &  0.0991& 0.0822& 0.1097& 0.0912& 0.1118& 0.0928\\
			&Improv. &\textcolor{myred}{ +40.76\%}& \textcolor{myred}{ +40.75\%}& \textcolor{myred}{ +34.11\%}& \textcolor{myred}{ +34.71\%}& \textcolor{myred}{ +23.95\%}& \textcolor{myred}{ +25.23\%}\\
			\bottomrule
		\end{tabular}
	}
\end{table*}

\subsubsection{Prior encoding} We hereby introduce a practical implementation of prior encoding based on three intuitive signals: popularity, quality, and hardness. For an enriched interaction \(\mathbf{x} = (u, i_{1:M}, j_{1:N})\), we define the exposure and suppression priors as follows:
\begin{align}
\pi^+_i &\propto \mathrm{rar}(i)^{\lambda_1^+} \cdot\mathrm{qual}(i) ^{\lambda_2^+} \cdot  \mathrm{hard}^+(i)^{\lambda_3^+},\label{eq:prior1} \\
\pi^-_j &\propto \mathrm{pop}(j)^{\lambda_1^-} \cdot\mathrm{bad\_qual}(j) ^{\lambda_2^-} \cdot  \mathrm{hard}^-(j)^{\lambda_3^-}\label{eq:prior2}.
\end{align}
Here, \(\mathrm{rar}(i) = 1 - \mathrm{pop}(i)\) measures the rarity of an item, and  \(\mathrm{pop}(i) \in [0,1]\) represents the log-normalized interaction count. In cases where ratings are available, we convert them into a bounded quality score \(\mathrm{qual}(j) \in (0,1)\) using mean-centered sigmoid normalization. Additionally, \(\mathrm{bad\_qual}(j) = 1 - \mathrm{qual}(j)\). For datasets without ratings (e.g., Gowalla and Yelp2018), this term is omitted. Hardness is evaluated bag-wise from model scores to identify samples that lie close to the decision boundary. Specifically, the hardness of a positive sample \(i\) is defined as: 
$\mathrm{hard}^+(i) = \mathrm{softmax} ((\bar{s}^+ - s_{i})/{\tau} )$,
and the hardness of a negative sample \(j\) is defined as: 
$\mathrm{hard}^-(j) = \mathrm{softmax} ( (s_{j} - \bar{s}^-)/{\tau} )$,
where \(\bar{s}^+\) and \(\bar{s}^-\) are the mean scores of the positive and negative samples in the bag, respectively.

Intuitively, \(\boldsymbol{\pi}^+\) encodes the exposure profile of rare, hard, high-quality items, while \(\boldsymbol{\pi}^-\) encodes the suppression profile of popular, hard, low-quality items. The prior exponents \(\{\lambda^\pm_{\cdot}\}\) provide an intuitive control interface: each weight  determines the strength of its corresponding signal (e.g., rarity), enabling the flexible encoding of tailored exposure or suppression profiles based on application requirements.

The popularity and quality signals are precomputed and stored in a buffer, while the hardness signal relies on similarity scores for items in the bag. These scores are not computed separately for hardness evaluation but are reused during both variational inference and loss computation. As a result, the additional computational overhead from prior encoding is minimal. Furthermore, the prior encoding does not require normalization, as the softmax in the variational posterior is invariant to constant scaling. What truly matters is the relative magnitude of these priors within bag.

\subsubsection{Hyper-parameter settings}
VarBPR was implemented in PyTorch. For fair comparison, all methods used the same backbone encoder and the same training protocol (including uniform random sampling of training pairs), and were optimized with Adam (learning rate $10^{-3}$).
All VarBPR-specific hyperparameters, including bag size $(M,N)$, regularization strength $(c_{\mathrm{pos}},c_{\mathrm{neg}})$, and prior exponents ${\lambda^\pm_{~\cdot}}$, are reported in the released code for reproducibility.

\subsection{Ranking Performance}
We first examine whether VarBPR outperforms strong baselines in top-$K$ ranking accuracy (\textbf{RQ1}).
First, using MF as the backbone, Table~\ref{tab:exp1} compares a range of training objectives. Popularity debiasing, denoising, and hard negative mining consistently improve over vanilla BPR. VarBPR achieves the best accuracy by integrating these effective mechanisms---hardness-aware learning via prior design, denoising via softened assignments, and popularity-aware debiasing via prior matching---within a unified variational objective.
Second, we replace the BPR loss with VarBPR in widely adopted backbones (MF~\cite{Koren:2009:Computer}, LightGCN~\cite{Xiangnan:2020:SIGIR}, and XSimGCL~\cite{10158930}) while keeping backbone hyperparameters unchanged. Table~\ref{tab:exp2} shows consistent gains across architectures, with relative improvements of up to 40\%.

\begin{figure}[!h]
	\centering
	\subfigure[Within-bag KL divergence.]
	{\includegraphics[width=0.492\columnwidth]{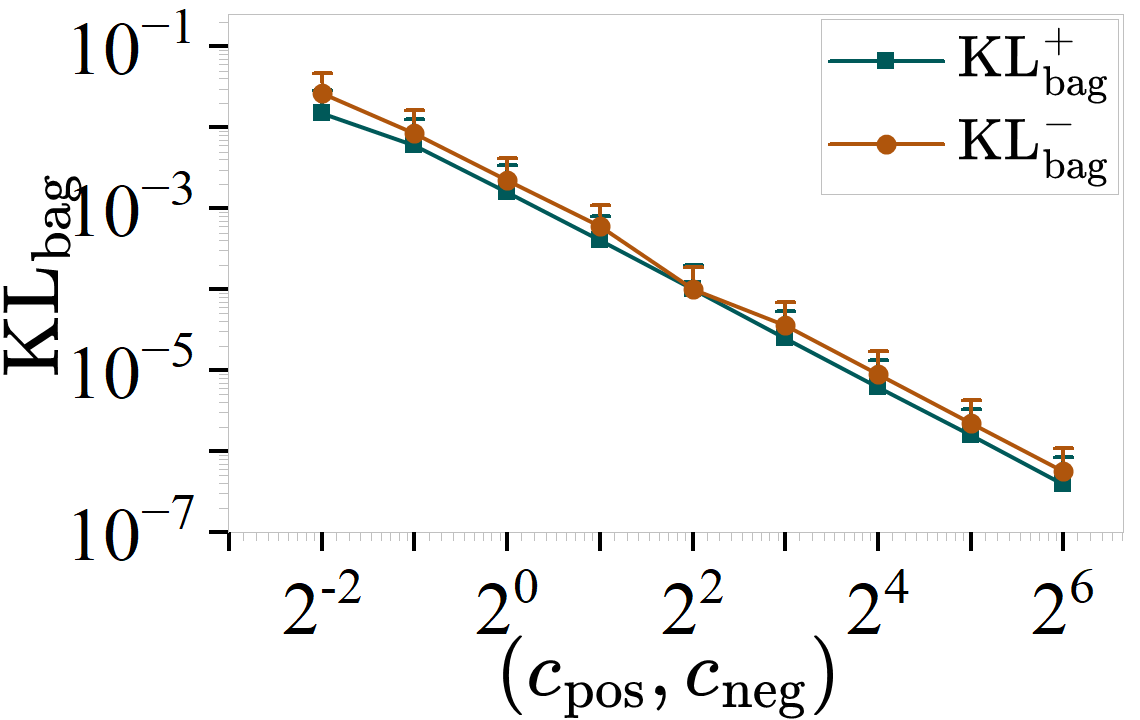}\label{fig:bagkl}}
	\subfigure[Global KL divergence.]
	{\includegraphics[width=0.492\columnwidth]{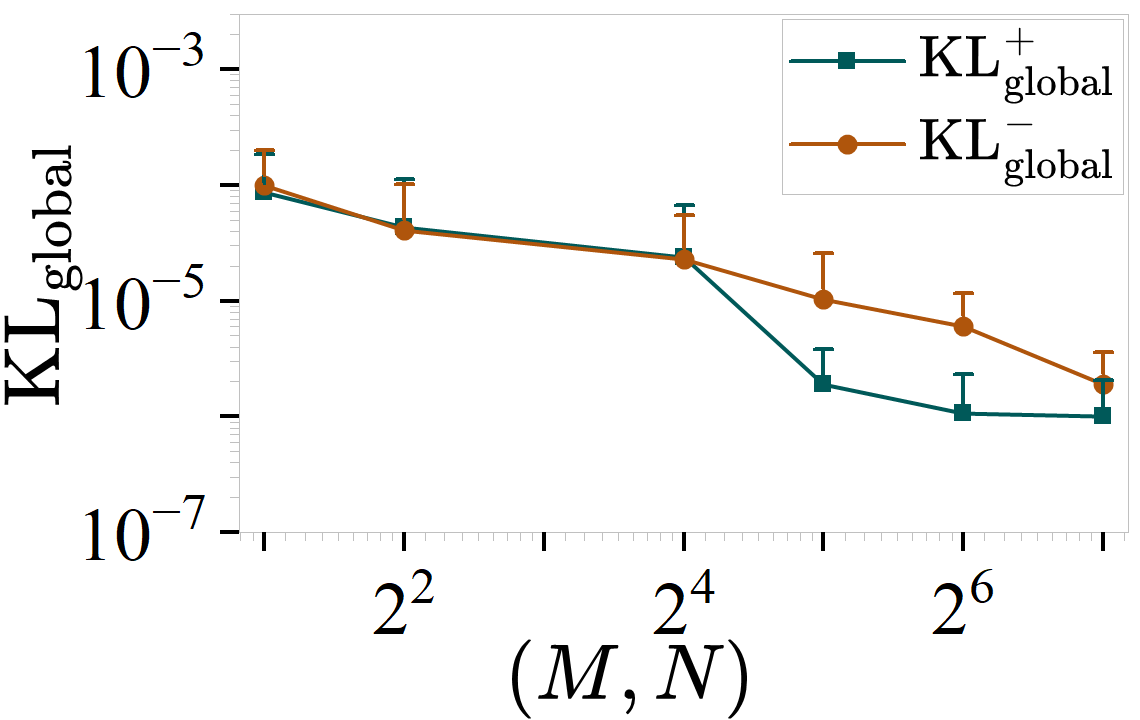}\label{fig:globalkl}}
    \caption{Variational posterior-prior alignment controlled by regulation strength $(c_{\mathrm{pos}},c_{\mathrm{neg}})$ and bag size $(M,N)$.}
\end{figure}
\begin{figure}[!h]
	\centering
	\subfigure[Long-tail oriented exposure.]
	{\includegraphics[width=0.492\columnwidth]{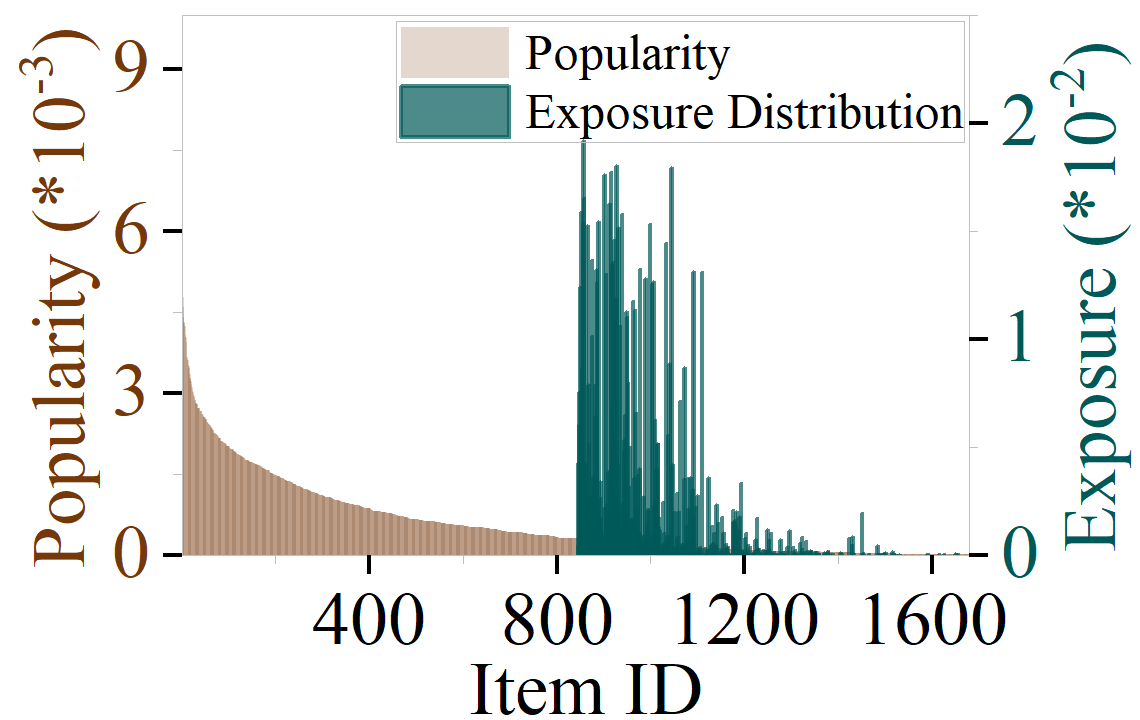}\label{fig:dirct1}}
	\subfigure[Quality oriented exposure.]
	{\includegraphics[width=0.492\columnwidth]{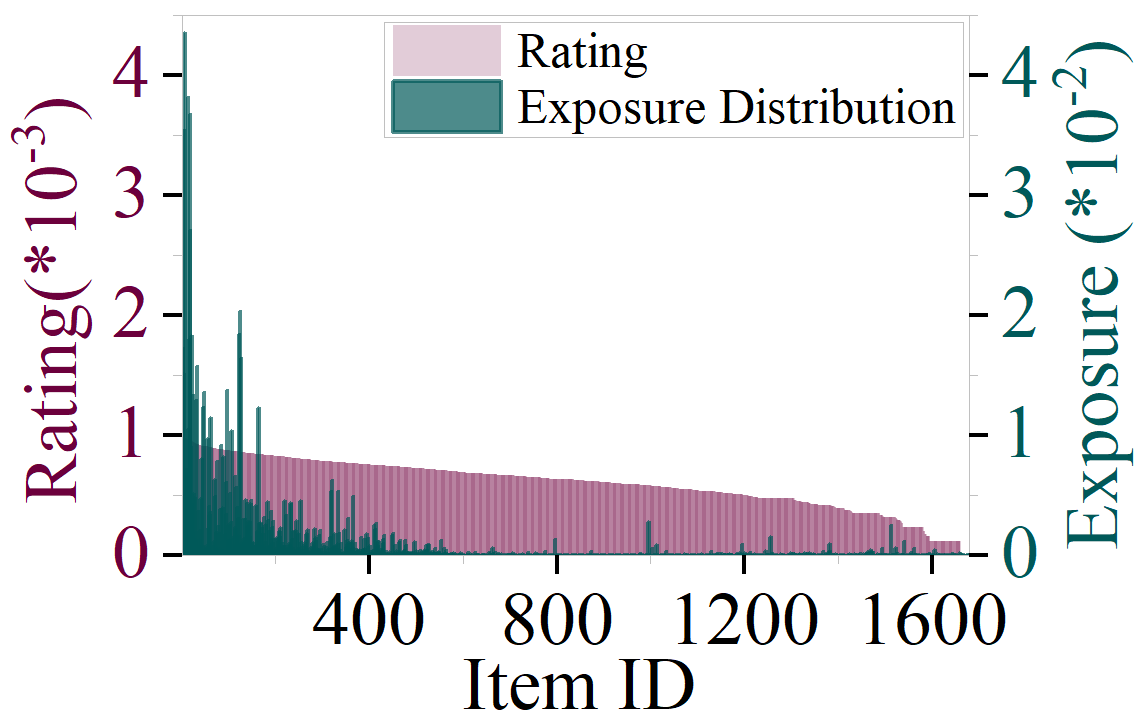}\label{fig:dirct2}}
    \caption{The prior exposure pattern determines the empirical exposure distribution under high-compliance settings, with two notable orientations: long-tail oriented (a) and high-quality oriented (b).}
\end{figure}
\begin{figure}[!h]
	\centering
	\includegraphics[width=0.49\textwidth]{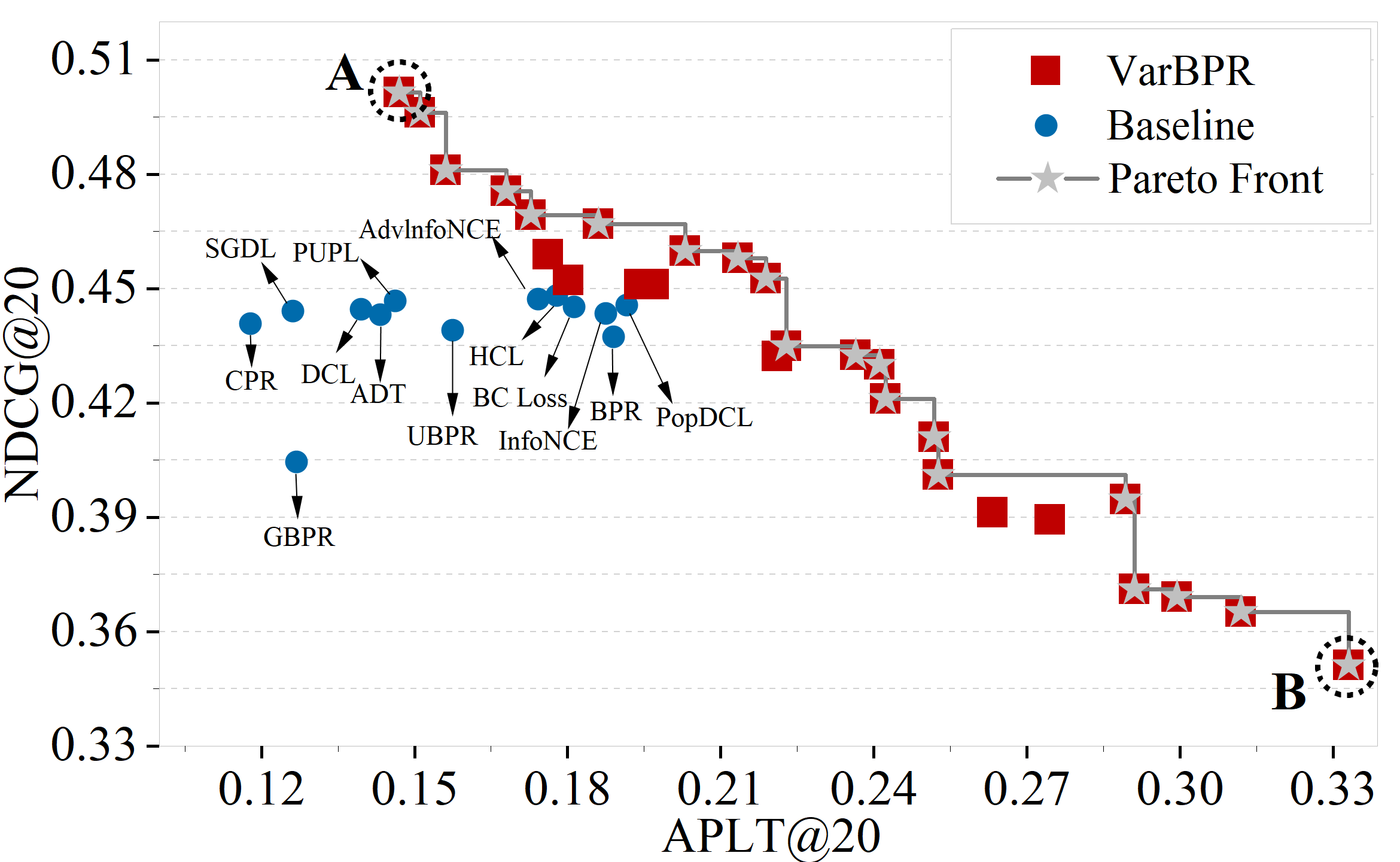}
	\caption{VarBPR enables a controllable long-tail exposure through a flexible direction-strength variational mechanism.}
	\label{fig:pareto}
\end{figure}
\begin{figure*}[!h]
	\centering
	\includegraphics[width=\textwidth]{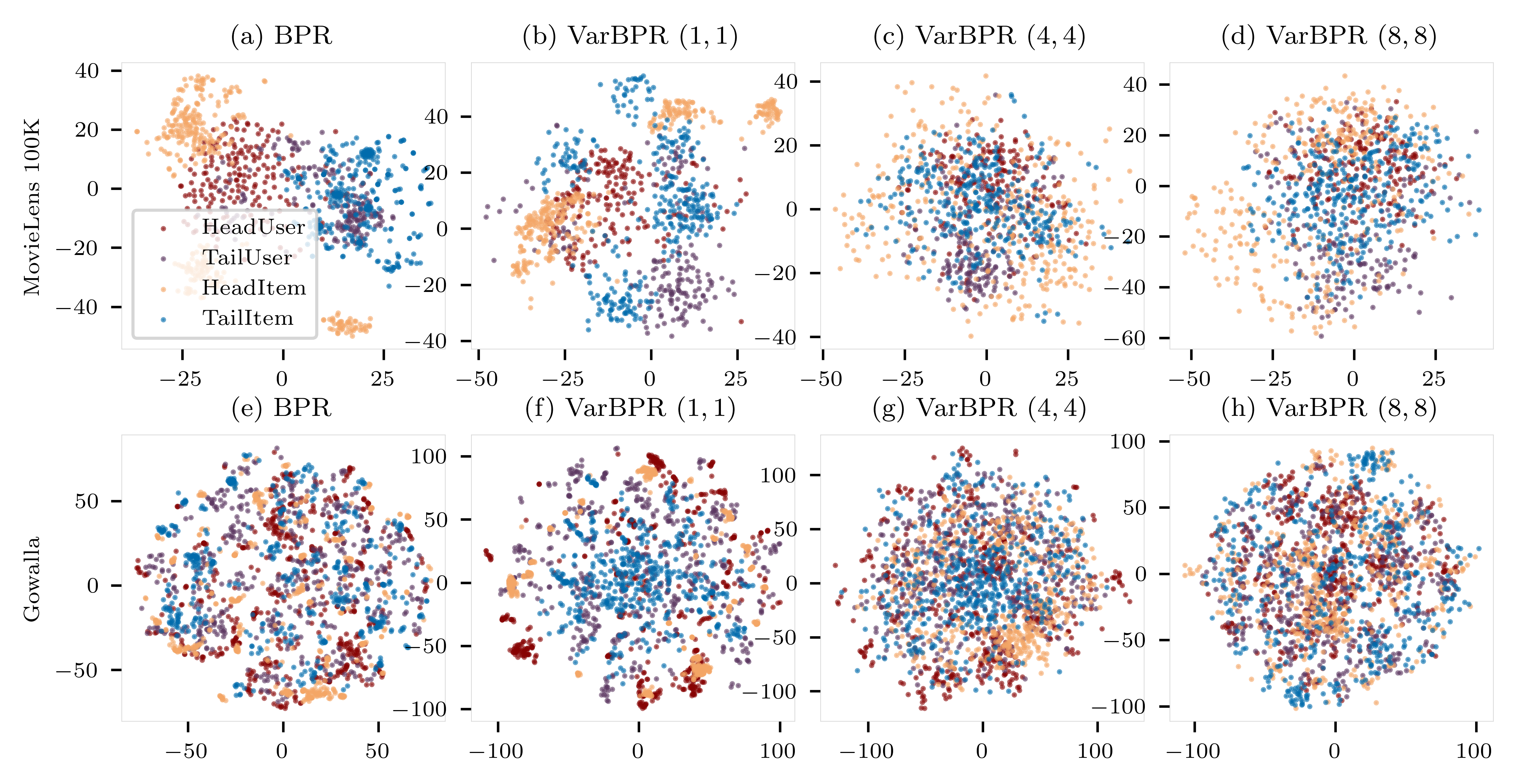}
	\caption{With an increase in the strength parameters $(c_{\mathrm{pos}}, c_{\mathrm{neg}})$ for the long-tail exposure oriented prior, the representations of head/tail items and users progressively become more uniform, accompanied by a controlled increase in long-tail exposure.}
	\label{fig:dist}
\end{figure*}
\subsection{Controllable Long-tail Exposure}
We demonstrate that VarBPR enables controllable long-tail exposure via its variational inference procedure (\textbf{RQ2}).
We decompose controllability into two orthogonal factors: (i) the target exposure profile encoded by the prior (direction), and (ii) posterior--prior compliance controlled by the inference hyperparameters $(c_{\mathrm{pos}}, c_{\mathrm{neg}}, M, N)$ (strength).

\textit{Inference hyperparameters $(c_{\mathrm{pos}}, c_{\mathrm{neg}}, M, N)$ control the strength of posterior--prior compliance.}
We quantify posterior--prior adherence at two scales.
\textbf{(i) Bag-level compliance.}
Fixing $M=N=4$, we vary $(c_{\mathrm{pos}},c_{\mathrm{neg}})$ and report
$\mathrm{KL}_{\mathrm{bag}}^{+}=\frac{1}{|\mathcal{S}|}\sum_{\mathbf{x}}\!\big[\mathrm{KL}(\boldsymbol{\alpha}(\mathbf{x})\|\boldsymbol{\pi}^{+}_{B}(\mathbf{x}))\big]$
and
$\mathrm{KL}_{\mathrm{bag}}^{-}=\frac{1}{|\mathcal{S}|}\sum_{\mathbf{x}}\!\big[\mathrm{KL}(\boldsymbol{\beta}(\mathbf{x})\|\boldsymbol{\pi}^{-}_{B}(\mathbf{x}))\big]$,
where $\boldsymbol{\pi}^{\pm}_{B}(\mathbf{x})$ is the prior restricted to the sampled items in $\mathbf{x}$ and renormalized.
Fig.~\ref{fig:bagkl} shows that $(c_{\mathrm{pos}},c_{\mathrm{neg}})$ provides a smooth handle on within-bag divergence. \textbf{(ii) Global compliance.}
Fixing $c_{\mathrm{pos}}=c_{\mathrm{neg}}=4$ and varying $(M,N)$, we report
$\mathrm{KL}_{\mathrm{global}}^{+}=\frac{1}{|\mathcal{U}|}\sum_{u}\!\big[\mathrm{KL}(\boldsymbol{\alpha}_{u}\|\boldsymbol{\pi}^{+}_{u})\big]$
and
$\mathrm{KL}_{\mathrm{global}}^{-}=\frac{1}{|\mathcal{U}|}\sum_{u}\!\big[\mathrm{KL}(\boldsymbol{\beta}_{u}\|\boldsymbol{\pi}^{-}_{u})\big]$.
Here, $\boldsymbol{\pi}^{+}_{u}$ (resp.\ $\boldsymbol{\pi}^{-}_{u}$) is the normalized prior on $\mathcal{I}^{+}_{u}$ (resp.\ $\mathcal{I}^{-}_{u}$), and
$\boldsymbol{\alpha}_{u},\boldsymbol{\beta}_{u}$ are user-level posteriors formed by pooling variational mass across bags (with zero fill for unsampled items) and renormalizing on $\mathcal{I}^{\pm}_{u}$.
Fig.~\ref{fig:globalkl} shows that bag size $(M,N)$ systematically modulates global variational posterior-prior adherence.
Overall, $(c_{\mathrm{pos}},c_{\mathrm{neg}})$ controls the \textit{strength} of adherence, while $(M,N)$ controls its \textit{coverage}, jointly shaping posterior-prior alignment.

\textit{Prior determines the exposure direction.}
Next, we verify that the prescribed priors $(\boldsymbol{\pi}^+,\boldsymbol{\pi}^-)$ determine which exposure pattern the posteriors compromise toward.
On the MovieLens-100K dataset, we adopt a high-compliance setting with $c_{\mathrm{pos}}=c_{\mathrm{neg}}=100$ and the largest feasible bag sizes $(M,N)$ (approximately covering $\mathcal{I}_u^+$ and $\mathcal{I}_u^-$). This configuration ensures that the variational inference is policy-driven, compelling the posterior to closely align with the prior over their respective supports.
Next, we consider two representative priors.
\textit{(1) Long-tail oriented exposure:} $\pi^+$ assigns unit weight to the bottom 50\% items by popularity (and zero otherwise), while $\pi^-$ assigns unit weight to the top 50\% items to further suppress head items.
\textit{(2) Quality-oriented exposure:} $\pi^+$ assigns unit weight to the top 50\% items by quality (average rating), while $\pi^-$ assigns unit weight to the bottom 50\% items to suppress low-quality items.
Fig.~\ref{fig:dirct1} compares the resulting Top-$20$ exposure profile against the popularity distribution (sorted in descending order), showing that exposure mass shifts markedly toward the bottom 50\% items.
Fig.~\ref{fig:dirct2} compares the Top-$20$ exposure profile against the quality distribution (sorted by average rating in descending order), showing a consistent shift toward the top 50\% items by quality.
Overall, the empirical exposure profile moves coherently with $\boldsymbol{\pi}^+$, with $\boldsymbol{\pi}^-$ providing an auxiliary suppressive constraint, confirming that the prior specifies the exposure direction.

\textit{Controllable long-tail exposure.}
In top-$K$ recommendation, improving ranking accuracy typically concentrates exposure on the most relevant items—often high-quality head items—leaving less exposure budget for the long tail. Consequently, there is an inherent trade-off between ranking accuracy and long-tail exposure, and a practical learning algorithm should provide an explicit control interface over long-tail exposure to meet diverse application needs.

Next, we show that VarBPR enables controllable long-tail exposure through a flexible direction-strength variational inference mechanism. \textbf{(i) Direction control}:
For exposure control, in the exposure prior (Eq.~\ref{eq:prior1}), we adjust the rarity component \(\lambda_1^+\) in steps of 0.2, increasing from 0 to 1, and the quality component \(\lambda_2^+\) in steps of 0.2, decreasing from 1 to 0. This transition shifts the prior exposure pattern from quality-driven to long-tail-driven.
For suppression control, in the suppression prior (Eq.~\ref{eq:prior2}), we adjust the popularity component \(\lambda_1^-\) in steps of 0.1, increasing from 0 to 0.5, and the low-quality component \(\lambda_2^-\) in steps of 0.1, decreasing from 0.5 to 0. This transition moves the prior suppression pattern from low-quality to popular-item suppression, with a fixed hardness factor of 0.5 to ensure stable gradients. \textbf{(ii) Strength control}:
We adjust the regularization parameters \(c_{\text{pos}}\) and \(c_{\text{neg}}\) in steps of 2, increasing from 2 to 10, progressively enhancing the posterior–prior adherence.

\begin{figure}[!t]
	\centering
	\subfigure[Training time.]
	{\includegraphics[width=0.492\columnwidth]{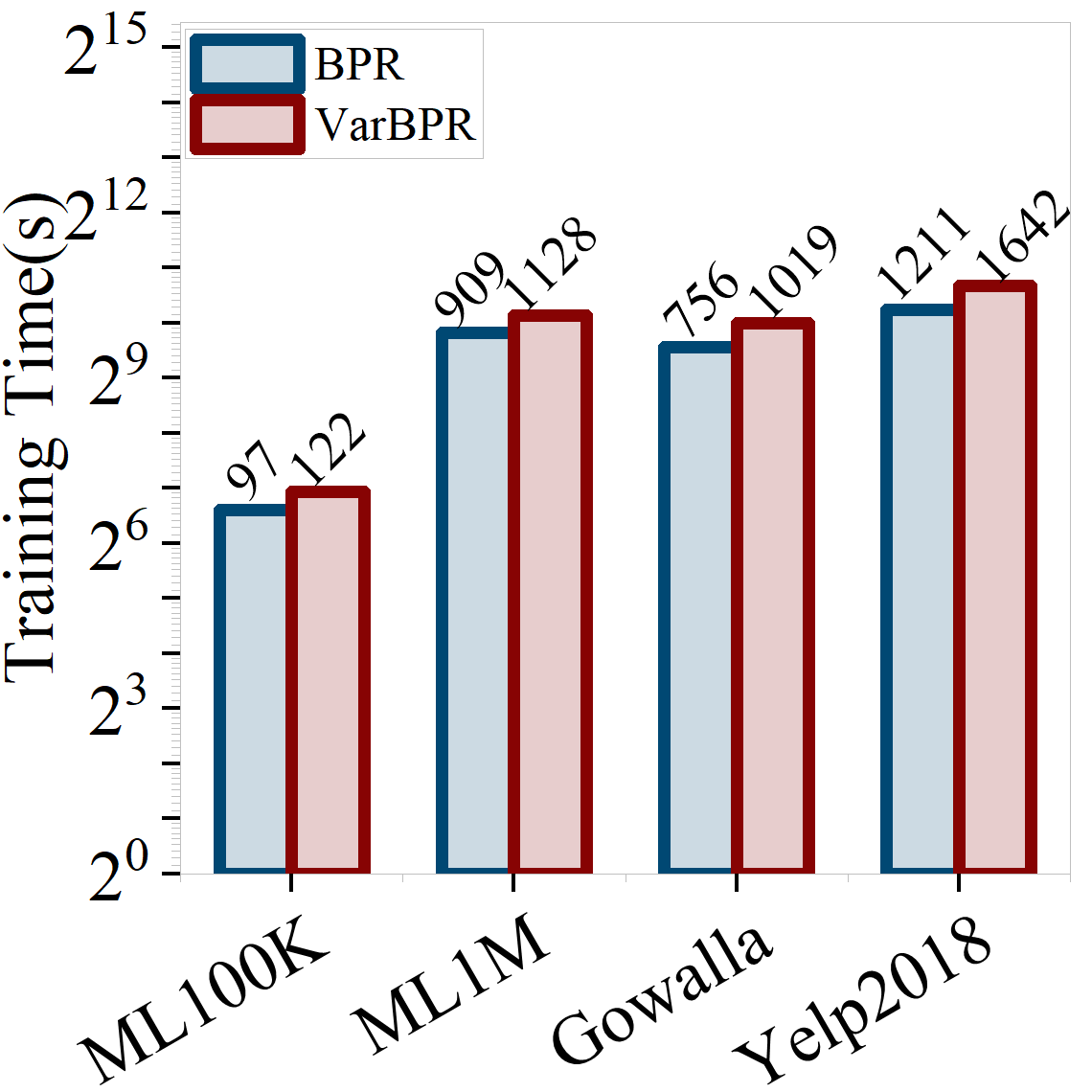}}
	\subfigure[Peak GPU memory usage.]
	{\includegraphics[width=0.492\columnwidth]{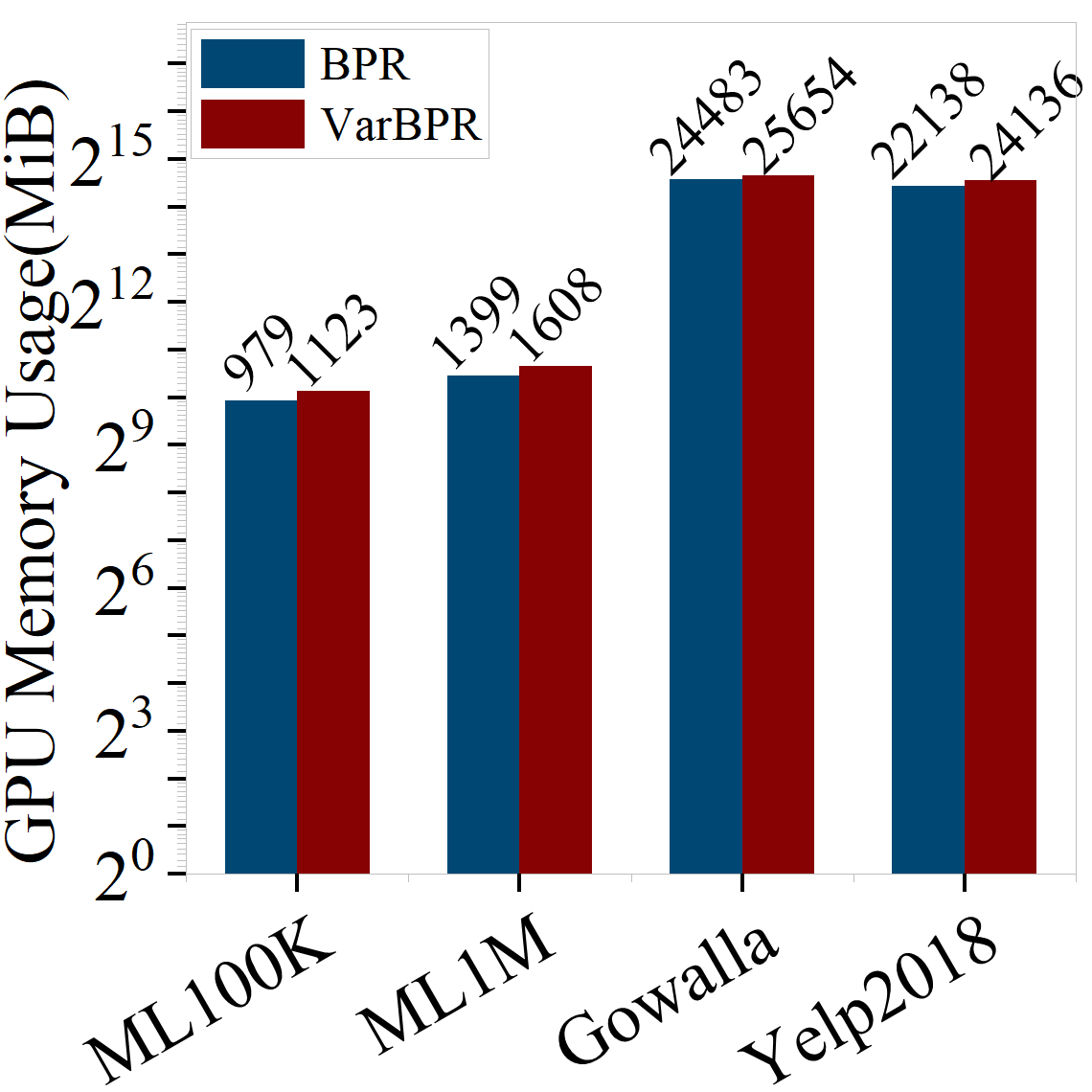}}
	\caption{VarBPR matches BPR efficiency and scales linearly.}\label{fig:sacle}
\end{figure}
Each unique combination of direction and strength parameters defines a distinct operating point on the exposure distribution and yields a corresponding pair of accuracy (NDCG) and long-tail exposure (APLT) metrics. Fig.~\ref{fig:pareto} provides empirical evidence from the MovieLens 100K dataset: VarBPR successfully traces a Pareto frontier comprising multiple optimal operating points, demonstrating its ability to achieve controllable long-tail exposure while achieving the best accuracy at the current level of long-tail exposure (Pareto optimal).
Fig.~\ref{fig:dist} further provides an direct evidence of controlled long-tail exposure: with a long-tail oriented prior ($\lambda_1^+=1, \lambda_1^-=0.5,\lambda_3^-=0.5$ in Eqs.~\eqref{eq:prior1}-\eqref{eq:prior2}), increasing $(c_{\mathrm{pos}}, c_{\mathrm{neg}})$ results in increasingly uniform feature representations across head/tail items and users, which is accompanied by a controlled increase in long-tail exposure.

\subsection{Efficiency and Scalability}\label{sec:time}
Next, we turn to evaluating the efficiency of VarBPR(\textbf{RQ3}). Taking MF as the fixed backbone model, we run 100 training epochs and employ exactly the same hyperparameter settings as reported in Table~\ref{tab:exp1}.
We report wall-clock training time and peak GPU memory against standard BPR in Fig.~\ref{fig:sacle}.
Experiments on MovieLens-100K and MovieLens-1M are conducted on a local machine with an NVIDIA GeForce RTX 3060 (12GB), while Yelp2018 and Gowalla are evaluated on a server with an NVIDIA Tesla V100-SXM2 (32GB).
Overall, VarBPR incurs comparable time and memory costs to BPR and scales linearly, consistent with its $\mathcal{O}(M{+}N)$ per-interaction complexity.

\subsection{Ablation on Inference Components}
\begin{table}[!ht]
	\caption{Ablation on key inference components.}\label{tab:ablation}
	\centering
	\resizebox{0.5\textwidth}{!}{
		\begin{tabular}{lllll}
			\toprule
			\multirow{2}{*}{Method} &   \multicolumn{2}{c}{MovieLens100K} &\multicolumn{2}{c}{MovieLens1M}  \\\cmidrule{2-5}
			& Recall@20 & NDCG@20 & Recall@20 & NDCG@20  \\\midrule
            \textit{w/o prior.} (uni. $\boldsymbol{\pi}^\pm $)   & 0.3389 $\textcolor{myblue}{_{4.9\%\downarrow}}$ & 0.4567 $\textcolor{myblue}{_{7.2\%\downarrow}}$ & 0.2333 $\textcolor{myblue}{_{2.3\%\downarrow}}$ & 0.4408 $\textcolor{myblue}{_{3.9\%\downarrow}}$ \\
            
            \textit{w/o VI} (uni. $\boldsymbol\alpha,\boldsymbol\beta$) & 0.2915 $\textcolor{myblue}{_{18.3\%\downarrow}}$& 0.4034 $\textcolor{myblue}{_{17.9\%\downarrow}}$& 0.2175 $\textcolor{myblue}{_{9.2\%\downarrow}}$& 0.4153 $\textcolor{myblue}{_{9.4\%\downarrow}}$ \\
            
           \textit{w/o plug-in} (ELBO)& 0.3631 $\textcolor{myred}{_{1.8.\%\uparrow}}$ & 0.5006 $\textcolor{myred}{_{1.8\%\uparrow}}$& 0.2428 $\textcolor{myred}{_{1.3\%\uparrow}}$& 0.4639 $\textcolor{myred}{_{1.2\%\uparrow}}$ \\
			VarBPR & \textbf{0.3566} & \textbf{0.4919} &  \textbf{0.2396}& \textbf{0.4586}\\
			\bottomrule
		\end{tabular}
	}
\end{table}
\begin{figure}[!h]
	\centering
	\subfigure[Jensen gap vs. margin variance.]
	{\includegraphics[width=0.492\columnwidth]{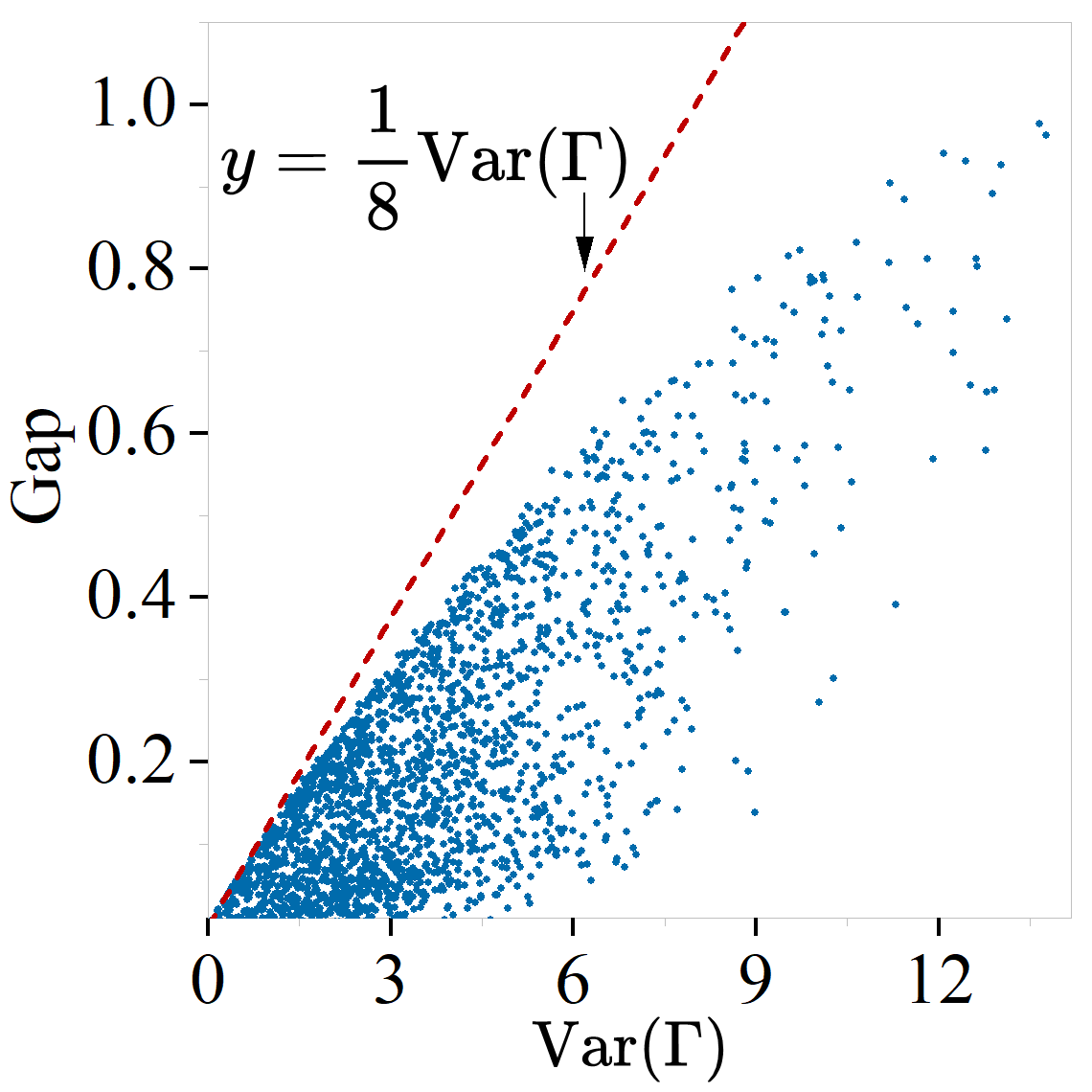}\label{fig:gap1}}
	\subfigure[Jensen gap across epochs.]
	{\includegraphics[width=0.492\columnwidth]{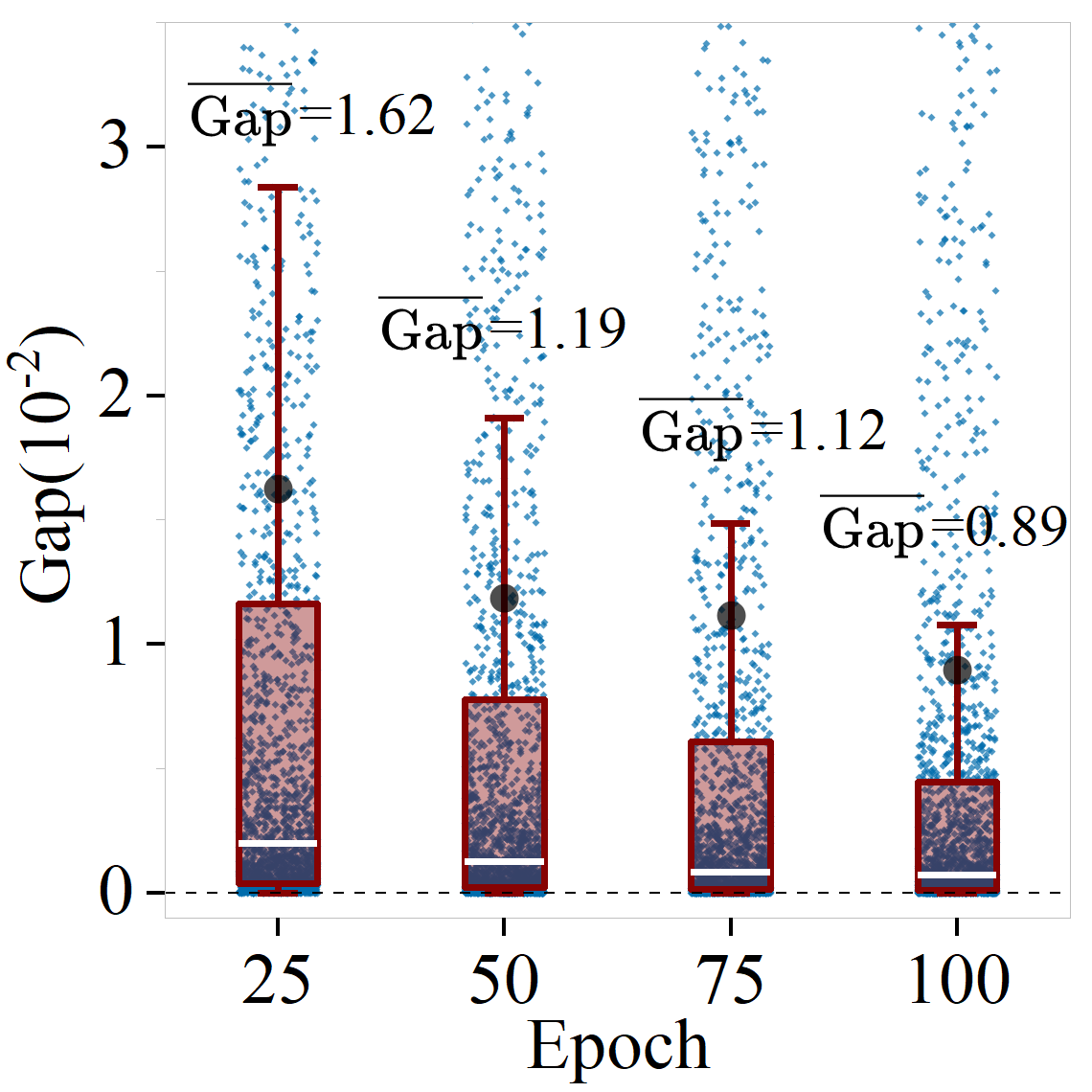}\label{fig:gap2}}
    \caption{Posterior compression error is controlled.}
\end{figure}
We conduct ablations in Table~\ref{tab:ablation} to isolate the contribution of each design choice in VarBPR (\textbf{RQ4}). \textit{w/o VI.} removes variational inference by replacing posteriors $(\boldsymbol\alpha,\boldsymbol\beta)$ with uniform weights. As a result, the centers reduce to simple mean pooling, leading to a significant drop in performance. This outcome underscores the critical role of variational inference in the model's effectiveness. The variational BPR posterior is the closed-form solution to a preference alignment variational problem that incorporates denoising and popularity debiasing regularization, which leads to richer and more reliable inference of intent signals.

\textit{w/o prior} removes prior by using uniform prior $\boldsymbol\pi^\pm$, resulting in a  consistent degradation. It also shows that a tailored exposure prior helps improve performance. However, we argue that the significance of the prior extends beyond just accuracy improvement. More importantly, it serves as a directional knob for controllable exposure, acting as an active module that allows practitioners to specify exposure objectives with particular goals in mind, such as category balance, quality, or fairness.

\textit{w/o plug-in} removes posterior summarization and optimizes the ELBO in Eq.~\eqref{eq:ideal-elbo2}. Its performance is slightly higher than the summarized VarBPR, indicating that plug-in approximation incurs a small loss while enabling linear-time training. Notably, the plug-in approximation is designed as a deliberate efficiency-performance trade-off. For flexibility, we also provide the ELBO variant of VarBPR in our open-source repository for users to choose based on needs.


\begin{figure}[!h]
	\centering
	\subfigure[5\% Artificial Noise]
    {\includegraphics[width=0.47\columnwidth]{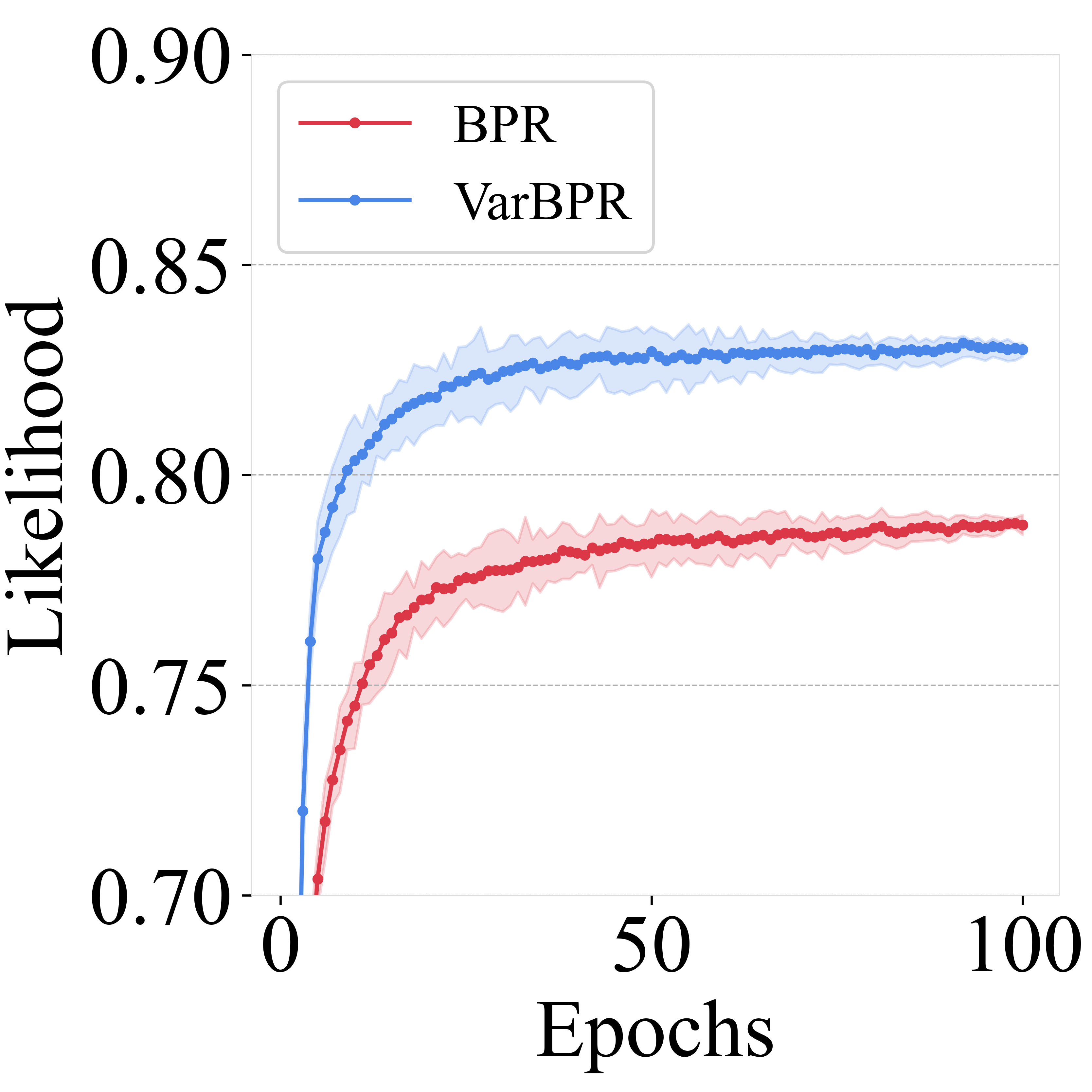}\label{fig:like5}}
	\subfigure[10\% Artificial Noise]{
		\includegraphics[width=0.47\columnwidth]{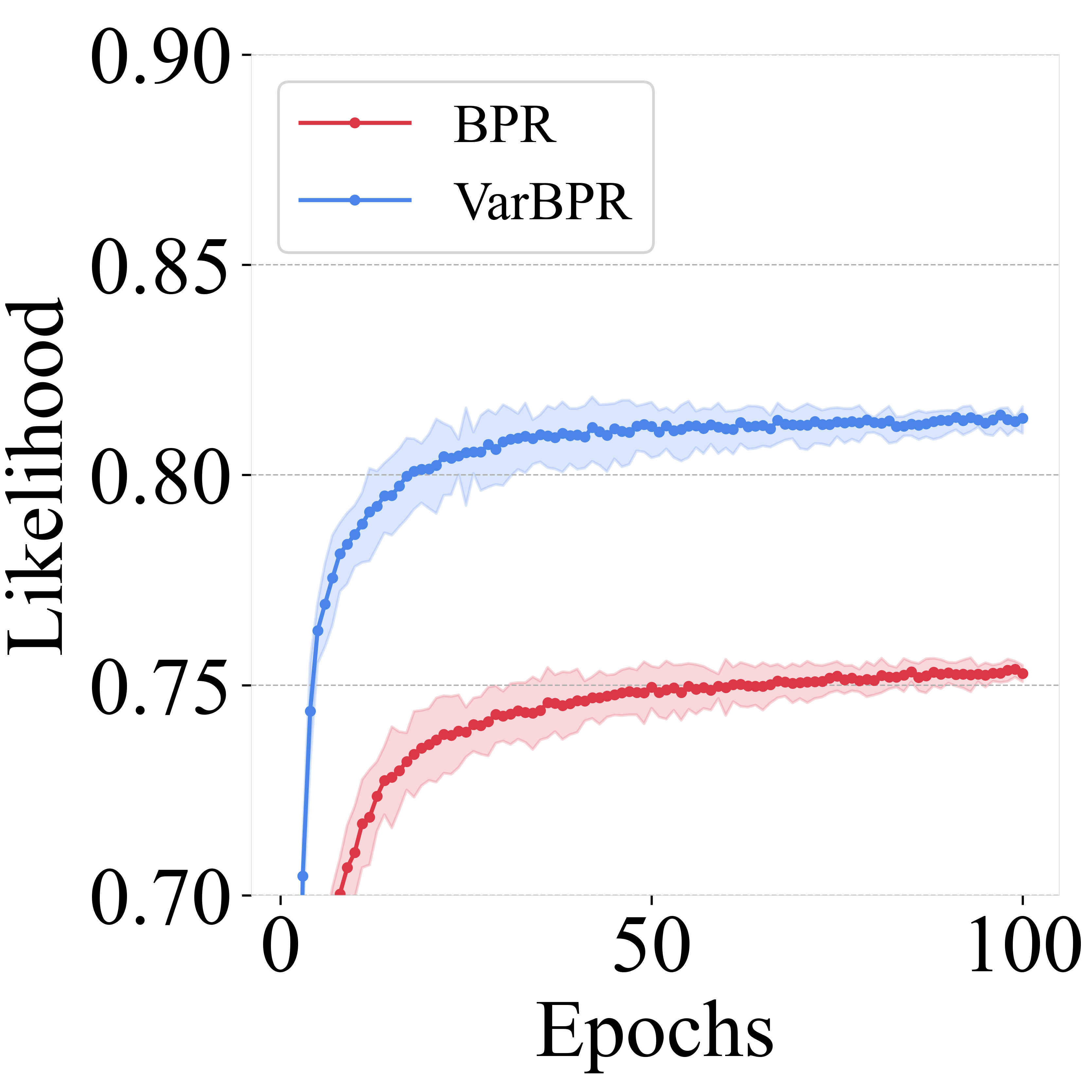}\label{fig:like10}}
	\caption{Likelihood vs. training epochs under increasing noise: the widening gap indicates that VarBPR degrades more slowly than BPR and is more robust to noise.}
	\label{fig:likelihood}
\end{figure}

\subsection{Validating Posterior Compression}
We validate the controllability of the plug-in Jensen gap (\textbf{RQ5}).
We train with VarBPR $\ell(\mathbb{E}[\Gamma])$ and simultaneously evaluate the uncompressed ELBO $\mathbb{E}[\ell(\Gamma)]$ on the same interactions, where $\Gamma$ is the posterior-induced random margin.
We define the compression gap as $\mathrm{Gap}\triangleq \ell(\mathbb{E}[\Gamma])-\mathbb{E}[\ell(\Gamma)]$.
Fig.~\ref{fig:gap1} samples $2^{11}$ interactions and plots Jensen $\mathrm{Gap}$ against $\mathrm{Var}(\Gamma)$; all points lie below $\tfrac18\,\mathrm{Var}(\Gamma)$, matching our theoretical bound.
Fig.~\ref{fig:gap2} shows per-epoch boxplots of Jensen $\mathrm{Gap}$ over all interactions, where the mean and median steadily decrease with training.
These results support the plug-in approximation as controlled, with a bounded compression error that steadily decreases over training.

\subsection{Robustness Study}
We next evaluate the robustness of VarBPR under noisy implicit feedback (\textbf{RQ6}). 
Since VarBPR maximizes a variational lower bound of the likelihood, it is expected to better preserve the underlying pairwise likelihood than standard BPR when observations are corrupted.
To test this, we trained the MF backbone on MovieLens100K using either BPR or VarBPR, and monitored the dataset-level likelihood at each epoch. Following~\cite{Bin:2023:ICDE}, we exclude training interactions, flip held-out test items to positive instances, and treat the remaining un-interacted items as negatives; the resulting likelihood is computed with the same sigmoid model as in BPR.
Then we inject $5\%$ and $10\%$ synthetic label noise into the training set, with the likelihood trajectories reported in Fig.~\ref{fig:like5} and~\ref{fig:like10}.VarBPR attains higher likelihood at both noise levels, and the widening likelihood gap as noise increases indicates that VarBPR degrades more slowly than BPR, i.e., it is less sensitive to noise.

\section{Conclusion}
This work establishes VarBPR as a tractable variational framework for controllable implicit-feedback ranking optimization, directly addressing three practical challenges in collaborative filtering: parameter estimation, noise perturbation, and popularity-driven exposure bias.
VarBPR recasts implicit-feedback pairwise learning as variational inference over discrete latent indexing variables that explicitly model noise and indexing uncertainty, and it separates training into two stages. In the variational inference stage, we develop a unified ELBO/regularization formulation that integrates preference alignment, denoising, and popularity debiasing, yielding closed-form posteriors with clear control semantics. Specifically, the prior shape specifies a target exposure pattern, while the temperature/regularization strength governs posterior--prior adherence, so exposure controllability emerges as an endogenous and interpretable outcome of VI rather than a post-hoc heuristic.
In the variational learning stage, we show that scalable learning is enabled by a posterior-compression objective that reduces the ELBO computation to linear, and we justify the resulting approximation via an explicit Jensen-gap bound. 
Theoretically, we provide interpretable generalization guarantees by identifying a structural error component induced by variational inference, decomposing it into controllable parts governed by the variational parameters, and characterizing the associated opportunity cost of prioritizing certain exposure patterns. This analysis yields a unified tuning principle centered on managing this cost and offers a concrete analytical lens for understanding and designing controllable recommender systems.
As a result, VarBPR supports recommender systems that are simultaneously effective, efficient, and deployable under application-specific exposure policies.
Future work may build on this opportunity-cost lens to develop principled methods for data-adaptive prior design and adaptive exposure-policy control.

\normalem
\bibliographystyle{IEEEtran}
\bibliography{ref}

\begin{IEEEbiography}
	[{\includegraphics[width=1in,height=1.25in,clip,keepaspectratio]{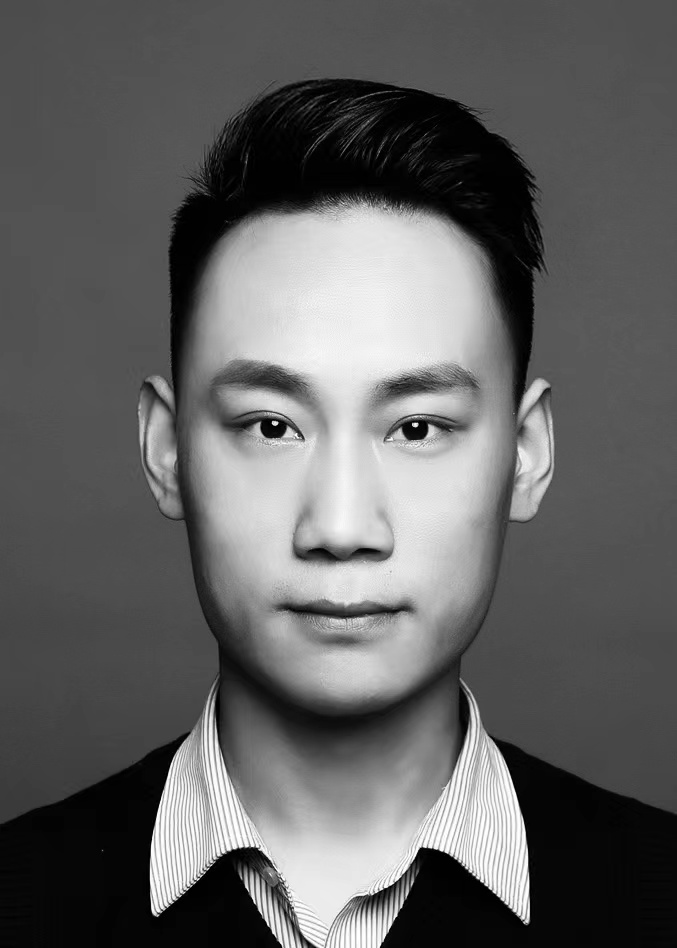}}]{Bin Liu} 
	received the Ph.D. degree from the School of Electronic Information and Communications, Huazhong University of Science and Technology, Wuhan, China, in 2023. He is currently affiliated with the School of Computing and Artificial Intelligence, Southwest Jiaotong University, Chengdu, China. He has published over 10 academic papers in top journals and conferences, including IEEE TPAMI, IEEE TKDE, IEEE ICDE. His research focuses on developing self-supervised and statistical machine learning techniques to advance information retrieval, ranking, and classification.
\end{IEEEbiography}

\begin{IEEEbiography}
[{\includegraphics[width=1in,height=1.25in,clip,keepaspectratio]{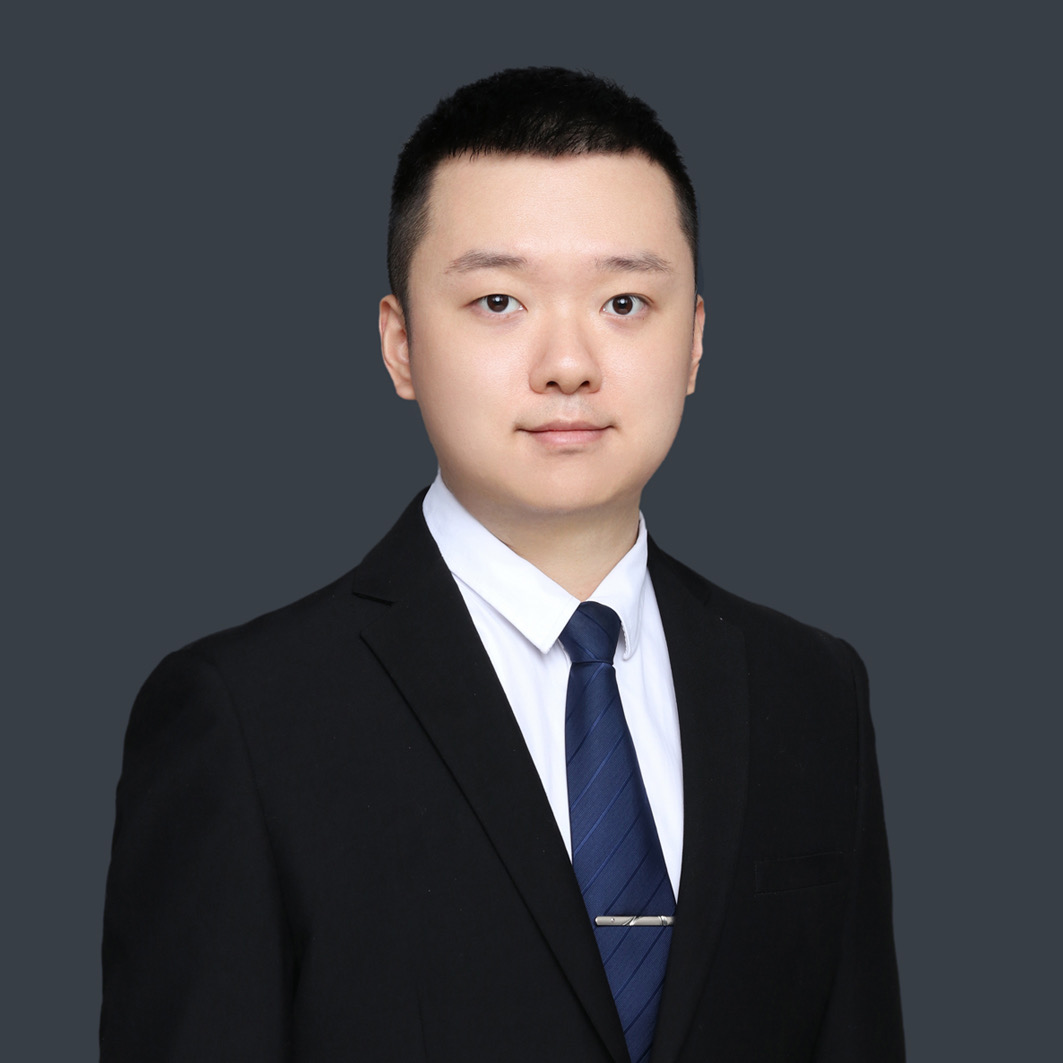}}]{Xiaohong Liu} (Member, IEEE)
received the Ph.D. degree in electrical and computer engineering from McMaster University, Canada, in 2021. He is currently an associate professor with the School of Computer Science, Shanghai Jiao Tong University, China. His research interests lie in computer vision and multimedia. He has published over 100 academic papers in top journals and conferences, and has been recognized
among the Stanford Top 2\% Scientists (2025), the Microsoft Research Asia StarTrack Scholars(2024), Chinese Government Award for Outstanding Self-financed Students Abroad (2021), and Borealis AI Fellowships (2020). He serves as an associate editor of the ACM Transactions on Multimedia Computing, Communications, and Applications.
\end{IEEEbiography}

\begin{IEEEbiography}
	[{\includegraphics[width=1in,height=1.25in,clip,keepaspectratio]{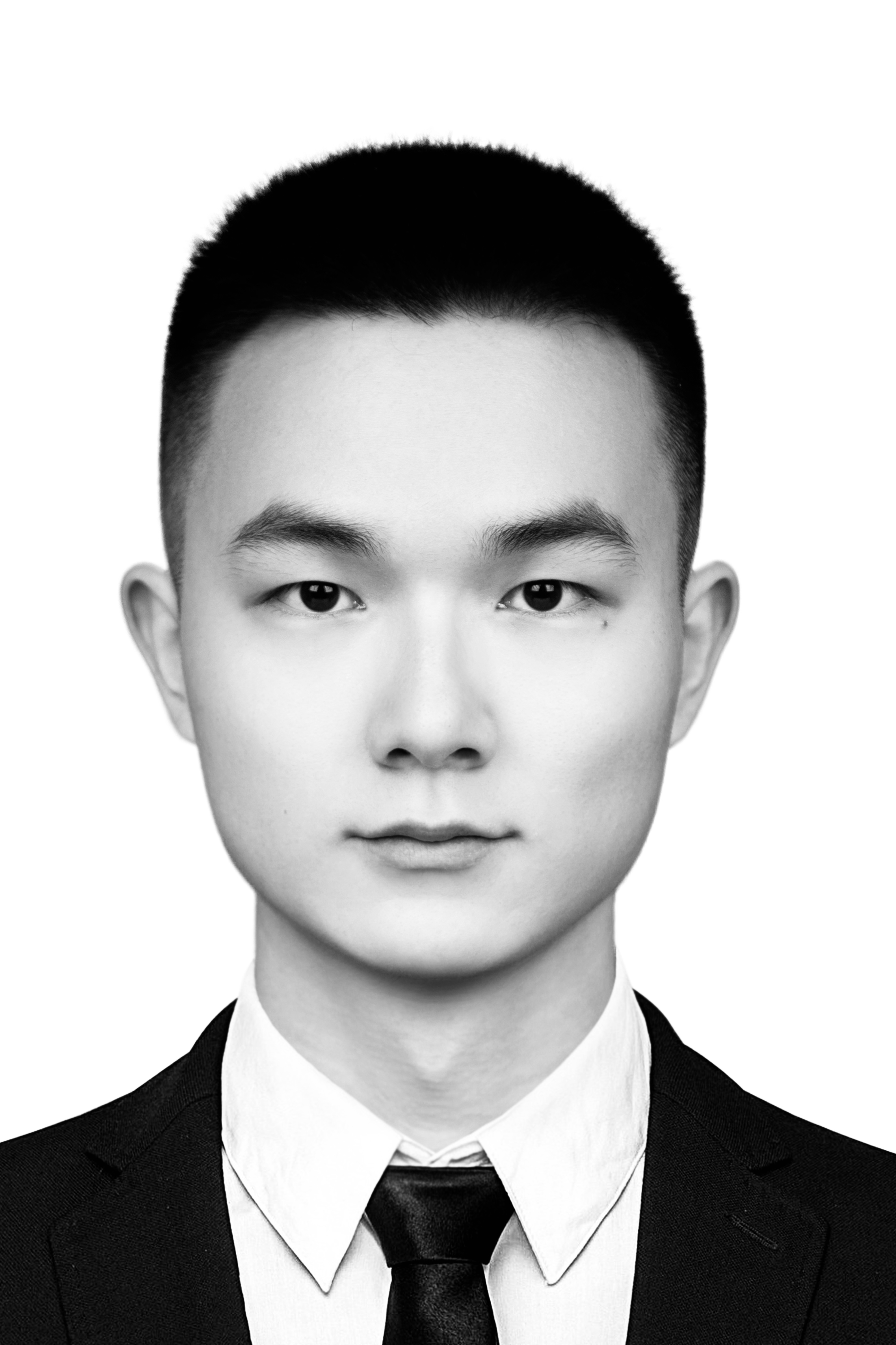}}]{Qin Luo} 
	received the B.S. degree from School of Electronic Information and Communications, Huazhong University of Science and Technology (HUST), Wuhan, China, in 2023. He is currently working toward the M.S. degree with the School of Electronic Information and Communications, Huazhong University of Science and Technology (HUST). He has published paper in IEEE Transactions on Knowledge and Data Engineering. His research interests include machine learning, artificial intelligence, and recommender systems.
\end{IEEEbiography}

\begin{IEEEbiography}  [{\includegraphics[width=1in,height=1.25in,clip,keepaspectratio]{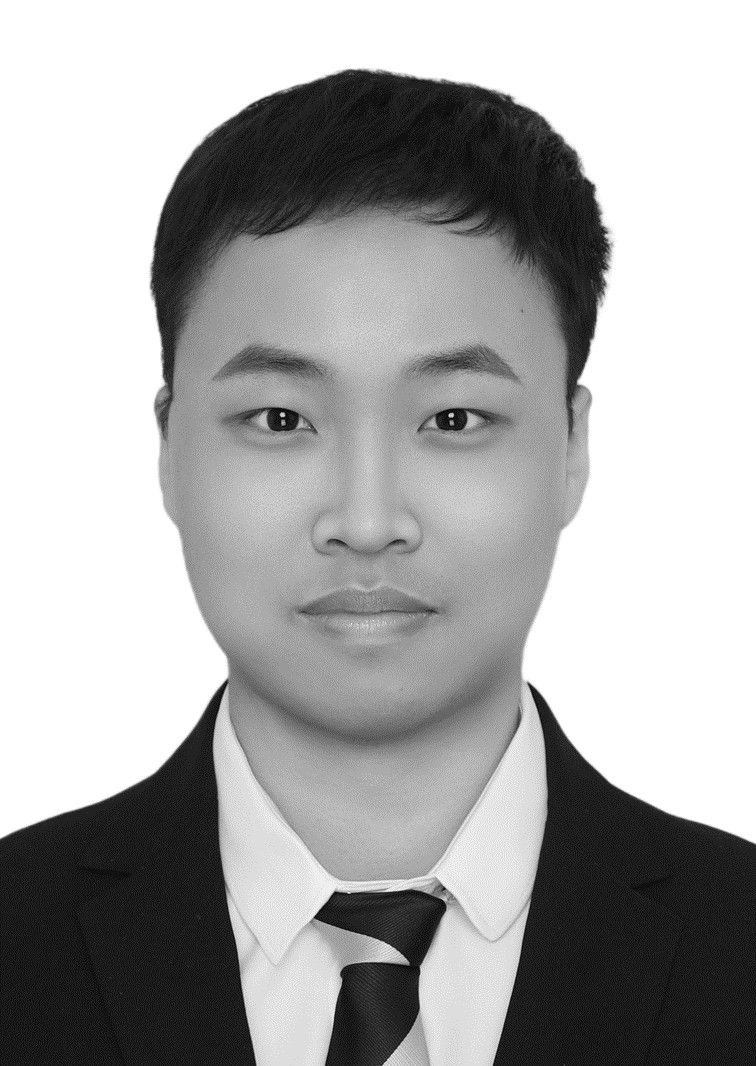}}] {Ziqiao Shang} received the M.Eng. degree from the School of Electronic Information and Communications, Huazhong University of Science and Technology, Wuhan, China, in 2024. His research interests primarily focus on facial-related computer vision tasks, including facial action unit detection and micro-expression recognition. He is currently pursuing the Ph.D. degree with the School of Intelligence Science and Technology, Nanjing University, Suzhou, China, under the supervision of Professor Lan-Zhe Guo.
\end{IEEEbiography}

\begin{IEEEbiography}[{\includegraphics[width=1in,height=1.25in,clip,keepaspectratio]{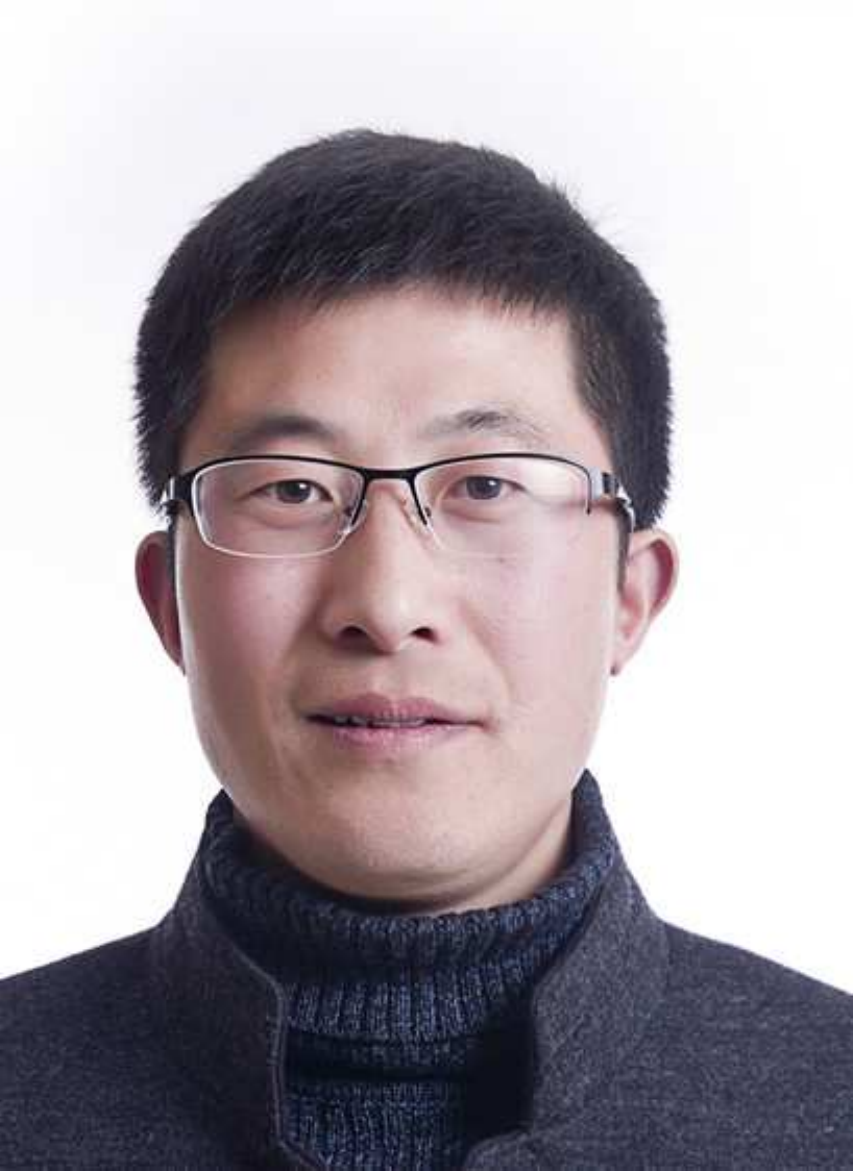}}]{Jielei Chu} 
	(Senior Member, IEEE) Jielei Chu received the
Ph.D. degrees in computer application and computer
science and technology from Southwest Jiaotong
University,Chengdu, China in 2020. He serves as Ed-
itorial Board Members for Scientific Reports. He has
published more than 40 scientific papers, including
IEEE Transactions on Pattern Analysis and Machine
Intelligence, IEEE Transactions on Knowledge and
Data Engineering, IEEE Transactions on Cybernetics
and IEEE Transactions on Multimedia. His research
interests include Deep Learning, Semi-supervised
Learning, Federated Learning and Brain-Inspired Intelligence.
\end{IEEEbiography}

\begin{IEEEbiography}[{\includegraphics[width=1in,height=1.25in,clip,keepaspectratio]{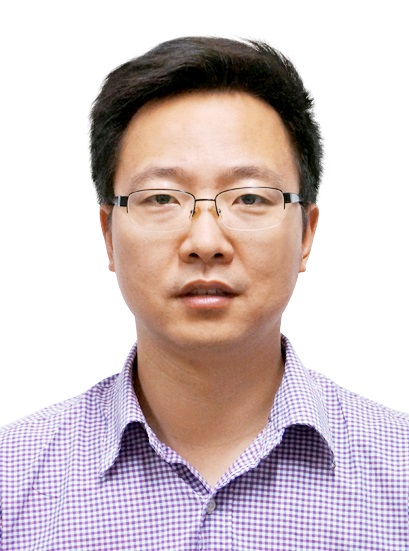}}]{Lin Ma} received the B.Sand M.S. degrees from the School of Information Science and Technology, Southwest Jiaotong University (SWJTU), SiChuan, China, in 2002 and 2007, respectively. He is currently a Associate Professor with the School of Transportation and Logistics, SWJTU. His research interests include machine learning, optimization method, intelligent transportation, multimodal freight transport.
\end{IEEEbiography}

\begin{IEEEbiography}
	[{\includegraphics[width=1in,height=1.25in,clip,keepaspectratio]{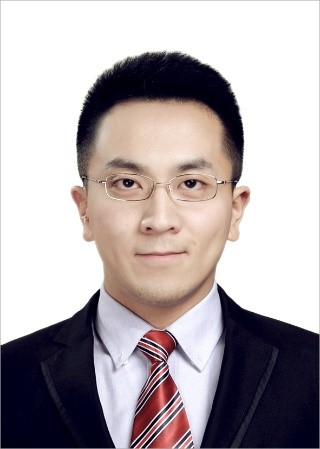}}]{Zhaoyu Li} is currently a senior engineer and the department head of CREC R\&D Center for Carbon Peaking and Carbon Neutrality of China Railway Group Limited. He is also a post-doctoral researcher in China Railway Group Limited and Southwest Jiaotong University, from 2023. He received his B.E. degree in civil engineering and M.E. degree in bridge and tunnel engineering from Southwest Jiaotong University, China, in 2012 and 2015, respectively, and the Ph.D. degree in Mechanical Engineering from the University of Auckland, New Zealand, in 2022. His research interests include machine learning, vibration energy harvesting, self-powered wireless sensor network, the CCUS and other intelligent systems and technologies for carbon neutrality.
\end{IEEEbiography}

\begin{IEEEbiography}  [{\includegraphics[width=1in,height=1.25in,clip,keepaspectratio]{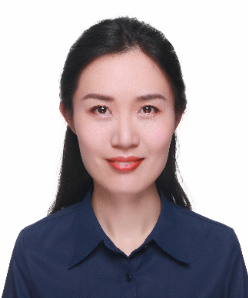}}] {Fei Teng} (Member, IEEE) received the BS and MS degrees from the Southwest Jiaotong University, China, in 2006 and 2008, respectively. She received her Ph.D. at Ecole Centrale Paris in France in 2011. She is currently a professor at the School of Computing and Artificial Intelligence, Southwest Jiaotong University. She is the reviewer of the IEEE Journal of Biomedical and Health Informatics, Information Sciences, and IEEE Transactions on Computers. She has authored or co-authored over 60 research papers in refereed journals and conferences. Her research interests include big data mining and service computing..
\end{IEEEbiography}

\begin{IEEEbiography}  [{\includegraphics[width=1in,height=1.25in,clip,keepaspectratio]{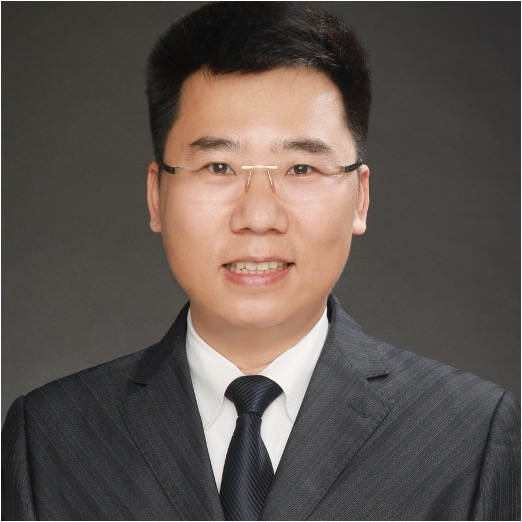}}] {Guangtao Zhai} (Fellow, IEEE) received the B.E. and M.E. degrees from Shandong Univer- sity,Shandong,China, in2001and 2004, respec- tively, and the Ph.D. degree from Shanghai Jiao Tong University, Shanghai, China, in 2009.From 2008 to2009, he was a Visiting Student with the Department of Electrical and Computer Engineering, McMaster University, Hamilton,ON,Canada,where he was a Postdoctoral Fellow from 2010to 2012. From2012 to 2013, he was a Humboldt Research Fellow with the Institute of Multimedia Communication and
	Signal Processing, Friedrich Alexander University of Erlangen-Nuremberg, Germany. He is currently a Professor at the Department of Electronics Engineering, Shanghai Jiao Tong University. He is a member of IEEE CAS VSPC TC and MSA TC. He has received multiple international and domestic research awards, including the Best Paper Award of IEEECVPR DynaVis Workshop 2020, the EasternScholarand Dawn program professorship of Shanghai, theNSFC Excellent Young Researcher Program, and the National Top Young ResearcherAward. He is serving as the Editor-in-Chief for Displays (Elsevier) and is on the Editorial Board for Digital Signal Processing(Elsevier) and Science China:Information Science.
\end{IEEEbiography}

\begin{IEEEbiography}  [{\includegraphics[width=1in,height=1.25in,clip,keepaspectratio]{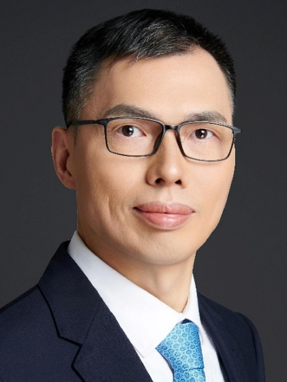}}] {Tianrui Li} (Senior Member, IEEE) received the BS, MS, and PhD degrees from Southwest Jiaotong University, Chengdu, China, in 1992, 1995, and 2002, respectively. He was a Postdoctoral Researcher with SCK•CEN, Belgium, from 2005 to 2006. He is currently a Professor and the Director of the Key Laboratory of Cloud Computing and Intelligent Techniques, Southwest Jiaotong University. He serves as Editor-in-Chief of Human-Centric Intelligent Systems, Editor of Information Fusion and Associate Editor of ACM Transactions on Intelligent Systems and Technology. He has authored or coauthored more than 500 research papers in refereed journals (e.g. IEEE TPAMI, IJCV, IEEE TKDE) and conferences (e.g. CVPR, ICCV, KDD). His research interests include big data, cloud computing, data mining, granular computing, and rough sets. He is a Fellow of IRSS and Senior Member of the ACM and IEEE.
\end{IEEEbiography}


\end{document}